\newif\ifmydraft\mydraftfalse
\newif\ifanonymous\anonymousfalse
\newif\iffull\fulltrue

\iffull
\documentclass[acmsmall,screen,nonacm\ifanonymous,anonymous\else,authorversion\fi]{acmart}
\else
\documentclass[acmsmall,screen]{acmart}
\fi

\usepackage{amsmath}
\usepackage{xspace}
\usepackage{xcolor}
\usepackage{listings}
\usepackage[scaled]{beramono}
\usepackage[normalem]{ulem}
\usepackage{booktabs}
\usepackage{hyperref}
\usepackage{wrapfig}
\usepackage[utf8]{inputenc}
\usepackage{microtype}
\usepackage{bcprules}
\usepackage{xargs}
\usepackage{stmaryrd}
\usepackage{upquote}
\usepackage[T1]{fontenc}
\usepackage{graphicx}
\usepackage[commentColor=violet]{algpseudocodex}
\usepackage{algorithm}
\usepackage{colortbl}
\usepackage{siunitx}

\newcommand{\souffle}{Souffl{\'e}\xspace}

\newcommand{\sep}{ \ensuremath{\!/\!} }
\newcommand{\stats}[4]{{#1}{#4}\sep{}{#2}{#4}\sep{}{#3}{#4}}
\newcommand{\speedups}[3]{\stats{#1}{#2}{#3}{\ensuremath{\times}}}

\newcommand{\hdr}[2]{%
\flushleft
\textbf{#1} \hfill {#2} \\
\centering
}

\newcommandx{\sidebyside}[5][1=.45,2=.45,3=b]{%
\begin{minipage}[#3]{#1\linewidth}
\centering
{#4}
\end{minipage}\hfill
\begin{minipage}[#3]{#2\linewidth}
\centering
{#5}
\end{minipage}%
}

\newcommandx{\threesidebyside}[7][1=.3,2=.3,3=.3,4=b]{%
\begin{minipage}[#4]{#1\linewidth}
\centering
{#5}
\end{minipage}\hfill
\begin{minipage}[#4]{#2\linewidth}
\centering
{#6}
\end{minipage}\hfill
\begin{minipage}[#4]{#3\linewidth}
\centering
{#7}
\end{minipage}%
}

\newcommandx{\foursidebyside}[8][1=.225,2=.225,3=.225,4=.225]{%
\begin{minipage}[b]{#1\linewidth}
\centering
{#5}
\end{minipage}\hfill
\begin{minipage}[b]{#2\linewidth}
\centering
{#6}
\end{minipage}\hfill
\begin{minipage}[b]{#3\linewidth}
\centering
{#7}
\end{minipage}\hfill
\begin{minipage}[b]{#4\linewidth}
\centering
  {#8}
\end{minipage}%
}

\theoremstyle{definition}
\newtheorem{axiom}{Axiom}

\newcommand{\cc}{C\texttt{++}\xspace}
\newcommand{\many}[1]{\ensuremath{\mathbf{#1}}}
\newcommand{\dl}[1]{\ensuremath{{#1}^{\mathrm{dl}}}}
\newcommand{\cpp}[1]{\ensuremath{{#1}^{\mathrm{cpp}}}}
\newcommand{\flg}[1]{\ensuremath{{#1}^{\mathrm{flg}}}}
\newcommand{\BNFALT}{\;\;|\;\;}
\newcommand{\prog}{\ensuremath{\mathrm{prog}}}
\newcommand{\horn}{\ensuremath{\mathrel{\mathsf{\mathord{:}\mathord{-}}}}}
\newcommand{\noteq}{\ensuremath{\mathrel{\mathsf{\mathord{!}\mathord{=}}}}}
\newcommand{\match}[2]{\ensuremath{\mathsf{match} ~ {#1} ~ \mathsf{with} ~ {#2}}}
\newcommand{\letin}[3]{\ensuremath{\mathsf{let} ~ {#1} = {#2} ~ \mathsf{in} ~ {#3}}}
\newcommand{\ite}[3]{\ensuremath{\mathsf{if} ~ {#1} ~ \mathsf{then} ~ {#2} ~ \mathsf{else} ~ {#3}}}
\newcommand{\transl}[1]{\ensuremath{\llbracket{#1}\rrbracket}}
\newcommand{\concat}{\ensuremath{\texttt{++}}}
\newcommand{\bk}{\textsf{\`{}}}
\newcommand{\bm}[1]{#1}
\newcommand{\happensbeforesym}{\ensuremath{\rightarrow}\xspace}
\newcommand{\happensbefore}[2]{\ensuremath{{#1}\happensbeforesym{#2}}}

\newcommand{\icfoot}[1]{\text{\lstinline[{basicstyle=\footnotesize\ttfamily\upshape}];#1;}}

\usepackage{enumitem}
\newlist{proofcases}{itemize}{3}
\setlist[proofcases]{leftmargin=2.25cm,style=sameline,parsep=3pt plus 1pt minus 1pt}

\ifmydraft
\fi

\AtBeginDocument{%
  }

\setcopyright{rightsretained}
\acmDOI{10.1145/3689754}
\acmYear{2024}
\acmJournal{PACMPL}
\acmVolume{8}
\acmNumber{OOPSLA2}
\acmArticle{314}
\acmMonth{10}
\acmSubmissionID{oopslab24main-p434-p}
\received{2024-04-03}
\received[accepted]{2024-08-18}

\begin{document}

\title{Making Formulog Fast: An Argument for Unconventional Datalog Evaluation}
\iffull
\subtitle{Extended Version}
\subtitlenote{This article extends one published in PACMPL~\citep{Bembenek2024Making} with technical appendices.}\fi       %

\author{Aaron Bembenek}
\orcid{0000-0002-3677-701X}
\affiliation{%
  \institution{University of Melbourne}
  \city{Parkville}
  \country{Australia}
}
\email{aaron.bembenek@unimelb.edu.au}

\author{Michael Greenberg}
\orcid{0000-0003-0014-7670}
\affiliation{%
  \institution{Stevens Institute of Technology}
  \city{Hoboken}
  \country{USA}
}
\email{michael@greenberg.science}

\author{Stephen Chong}
\orcid{0000-0002-6734-5383}
\affiliation{%
  \institution{Harvard University}
  \city{Cambridge}
  \country{USA}
}
\email{chong@seas.harvard.edu}

\begin{abstract}
  With its combination of Datalog, SMT solving, and functional programming, the language Formulog provides an appealing mix of features for implementing SMT-based static analyses (e.g., refinement type checking, symbolic execution) in a natural, declarative way.
  At the same time, the performance of its custom Datalog solver can be an impediment to using Formulog beyond prototyping---a common problem for Datalog variants that aspire to solve large problem instances.
  In this work we speed up Formulog evaluation, with some surprising results: while 2.2$\times$ speedups can be obtained by using the conventional techniques for high-performance Datalog (e.g., compilation, specialized data structures), the big wins come by abandoning the central assumption in modern performant Datalog engines, semi-naive Datalog evaluation.
  In the place of semi-naive evaluation, we develop eager evaluation, a concurrent Datalog evaluation algorithm that explores the logical inference space via a depth-first traversal order.
  In practice, eager evaluation leads to an advantageous distribution of Formulog's SMT workload to external SMT solvers and improved SMT solving times: our eager evaluation extensions to the Formulog interpreter and \souffle's code generator achieve mean 5.2$\times$ and 7.6$\times$ speedups, respectively, over the optimized code generated by off-the-shelf \souffle on SMT-heavy Formulog benchmarks.

  All in all, using compilation and eager evaluation (as appropriate), Formulog implementations of refinement type checking, bottom-up pointer analysis, and symbolic execution achieve speedups on 20 out of 23 benchmarks over previously published, hand-tuned analyses written in F$^\sharp$, Java, and \cc, providing strong evidence that Formulog can be the basis of a realistic platform for SMT-based static analysis.
  Moreover, our experience adds nuance to the conventional wisdom that traditional semi-naive evaluation is the one-size-fits-all best Datalog evaluation algorithm for static analysis workloads.

\end{abstract}

  


\maketitle

\section{Introduction}

Through combining Datalog, SMT solving, and functional programming, Formulog~\cite{formulog} aims to be a domain-specific language for declaratively implementing SMT-based static analyses such as refinement type checking and symbolic execution.
Formulog-based analyses achieve speedups over (non-Datalog) reference implementations on several case studies thanks to high-level optimizations like automatic parallelization and goal-directed evaluation; nonetheless, the prototype Formulog interpreter is not a heavily optimized system, and on pure Datalog programs it can be $\sim\!7\times$ slower than the industrial-strength Datalog solver \souffle~\citep{souffle1,souffle2}.
This performance gap is hardly surprising, as building a performant Datalog system is a challenging endeavor---involving multicore programming, specialized data structures, and subtle engineering---and \souffle has benefited from multiple person-years of engineering investment, as well as industry support.

The engineering gap between Formulog and an optimized Datalog system like \souffle has limited Formulog's use in practice.
For example, \citet{smaragdakis2021symbolic} conclude that Formulog's custom Datalog engine will not scale to the problems targeted by their novel ``symvalic'' analysis; this is despite the fact that Formulog has the right language features, and ``an evolved implementation [of Formulog] that will seamlessly combine Datalog rules and symbolic reasoning with high performance will be an ideal platform for symvalic analysis in the future.''
Formulog may have been a success from the perspective of language design, but its potential is only partially realized if the performance of its prototype implementation prevents interested clients from using it.

Moreover, performance concerns do not apply solely to Formulog: indeed, any novel Datalog variant that aspires to solve industrial-size problem instances faces the challenge of closing the sizable gap between research artifact and high-performance system.

\smallskip

Fortunately, over the past decade, a recipe has emerged for fast and scalable Datalog evaluation, based on techniques pioneered by \souffle: compile the Datalog program to imperative code that specializes the semi-naive Datalog evaluation algorithm~\citep{bancilhon1986naive} to that source program~\citep{souffle1,souffle2}, and link in optimized data structures~\citep{souffleindex,soufflebtree,souffletrie,Jordan2022Specializing}.
This is the general approach taken by recent high-performance Datalog engines (beyond \souffle) that target static analysis applications.
For example, Ascent~\citep{ascent} uses macros to compile embedded Datalog code to Rust code implementing semi-naive evaluation for that program, and then \textsc{Byods}~\citep{Sahebolamri2023Bring} adds the ability for users to link in custom data structures ideal for the particular workload; Flan~\citep{Abeysinghe2024Flan} uses a metaprogramming framework to specialize both a semi-naive Datalog interpreter and the data structures used by the interpreter to the Datalog program being evaluated.

\citet{pacak2022functional} sensibly propose ``frontend compilation'' as an approach to making novel Datalog-inspired languages performant: compile them to existing high-performance Datalog engines that already implement the traditional recipe for fast and scalable Datalog evaluation.
They contrast frontend compilation with two existing classes of approaches to making Datalog systems more expressive or usable: ``frontend-first'' approaches, which start by adding language features to Datalog and then build a custom backend to support them (here they explicitly use Formulog as an example), and ``backend-first'' approaches, which start with existing Datalog infrastructure and add new features on top of it.
They argue that frontend compilation achieves the best of both worlds, by allowing expressive frontend design, while reusing existing Datalog infrastructure.

Our first attempt to speed up Formulog evaluation follows the frontend compilation approach of \citet{pacak2022functional}.
We develop a compiler from the Datalog fragment of Formulog to \souffle, which in turn generates \cc code; this code is linked with the \cc code generated directly by our compiler for the non-Datalog features of Formulog (the functional fragment and SMT runtime).
The generated code is arithmetic mean $2.2\times$ faster than the baseline Formulog interpreter on Formulog benchmarks (min \sep median \sep max: \speedups{1.0}{1.7}{4.3}); the approach works particularly well for programs (like a bottom-up Java pointer analysis~\cite{bottomupPointsto}) where most computational work happens in the Datalog fragment (arith mean \sep min \sep median \sep max: 3.2$\times$ \sep 1.3$\times$ \sep 3.4$\times$ \sep 4.3$\times$).

\smallskip

Given the engineering effort put into \souffle, this speedup is not surprising; what is surprising is that one can do better, with little engineering effort.
With the addition of just 200 lines of code, the prototype Formulog interpreter beats the optimized \cc code generated by \souffle with an arithmetic mean 5.2$\times$ speedup on non-trivial SMT-heavy benchmarks (min \sep median \sep max: \speedups{0.41}{1.2}{31}), such as sophisticated refinement type checking~\cite{dminor} and KLEE-style symbolic execution~\cite{klee}, where there are up to hundreds of thousands of SMT calls and tens of millions of derived tuples.
The code generated by \souffle has many potential advantages over the Formulog interpreter (which is written in Java and uses relatively naive data structures);
the interpreter's trick is to abandon one of the key assumptions underlying modern high-performance Datalog engines, semi-naive evaluation.

Semi-naive evaluation~\cite{bancilhon1986naive} has been the standard Datalog evaluation algorithm for decades, and is the backbone of virtually all modern Datalog engines (especially those that target static analysis applications)~\citep{Ketsman2022Modern}.
Semi-naive evaluation cuts down on the number of redundant inferences made while evaluating Datalog rules, but it has a crucial limitation when it comes to Formulog evaluation: it enforces a breadth-first search (BFS) over the logical inference space, deriving inferences of proof height one, followed by inferences of proof height two, three, etc.
In the context of Formulog, the order of logical inferences determines how the SMT queries that arise in the course of Datalog rule evaluation are distributed amongst external SMT solvers; thus, improved Formulog performance can be achieved by making inferences in an order that leads to faster external SMT solving (whether thanks to improved opportunities for incremental SMT solving \citep{Een2003Temporal,formulog_incremental_smt} or just a more favorable distribution of the SMT workload).
Here, the BFS of semi-naive evaluation can be suboptimal.

In place of traditional semi-naive evaluation, we develop {\em eager evaluation}, a Datalog evaluation algorithm that uses work-stealing-based parallelism~\citep{mohr1990lazy,arora1998thread,blumofe1999scheduling} to perform a quasi-depth-first search (DFS) of the logical inference space.%
\footnote{
  The original Formulog paper~\citep{formulog} uses a preliminary version of eager evaluation for one of the case studies; however, it does not describe it nor evaluate it with respect to semi-naive evaluation. This paper fills those gaps.
}
Instead of batching work into explicit rounds of evaluation (as in semi-naive evaluation), eager evaluation eagerly pursues the consequences of the most recent derivations.
While this lack of batching means that eager evaluation can sometimes make more redundant derivations than semi-naive evaluation, in practice it is a good evaluation algorithm for Formulog: eager evaluation leads to advantageous distributions of the SMT workload across external SMT solvers and thus less time spent in SMT solving, a key determiner of overall performance for SMT-heavy Formulog programs.

We extend \souffle with proof-of-concept support for generating \cc code that performs eager evaluation.
On SMT-heavy benchmarks, the generated code achieves an arithmetic mean 1.8$\times$ speedup (min\sep{}median\sep{}max: \speedups{0.93}{1.3}{4.1}) over the Formulog interpreter running in eager evaluation mode.
This represents an arithmetic mean 8.8$\times$ speedup over the baseline Formulog interpreter (min\sep{}median\sep{}max: \speedups{0.73}{4.8}{38}).
As our modifications to the \souffle codebase amount to merely 500 additional lines of code, our experience demonstrates the feasibility of using existing infrastructure to build Datalog engines that perform eager evaluation.
This is a key strength of the algorithm, as a novel Datalog evaluation technique is unlikely to be used in practice if it is incompatible with the substantial research and sophisticated, labor-intensive engineering that has gone into making existing Datalog systems like \souffle performant.

\begin{table}
  \caption{
    Compilation and eager evaluation lead to substantial arithmetic mean speedups over the baseline Formulog interpreter.
    The headings ``+eager (interpret)'' and ``+eager (compile)'' denote using interpreted and compiled eager evaluation modes, respectively, for SMT-heavy benchmarks, and the base compiler (to off-the-shelf \souffle) for the other benchmarks.
    Eager evaluation is key to getting top performance on SMT-heavy benchmarks, where {\em interpreted} eager evaluation actually beats {\em compiled} semi-naive code.
  }\label{tab:summary}
  \begin{tabular}{lrrr}
    \toprule
    Benchmarks     & Base compiler & +eager (interpret) & +eager (compile) \\
    \midrule
    SMT-heavy (13) & 1.5$\times$   & 5.7$\times$        & 8.8$\times$      \\
    All (23)       & 2.2$\times$   & 4.6$\times$        & 6.4$\times$      \\
    \bottomrule
  \end{tabular}
\end{table}

\smallskip

Overall, our compiler from Formulog to \souffle---set to use our eager evaluation extension for SMT-heavy benchmarks---provides an arithmetic mean speedup of 6.4$\times$ over the baseline Formulog interpreter (min\sep{}median\sep{}max: \speedups{0.74}{3.4}{38}; see summary in Table~\ref{tab:summary}), and beats all interpreter modes on 22 out of 23 benchmarks.
This translates to competitive performance relative to previously published, hand-tuned analyses written in F$^\sharp$, Java, and \cc: on 20 out of 23 benchmarks, some Formulog mode is the fastest; on synthetic symbolic execution benchmarks, compiled Formulog's eager evaluation mode achieves an arithmetic mean 100$\times$ speedup (min\sep{}median\sep{}max: \speedups{7.0}{12}{870}) over the symbolic execution tool KLEE~\citep{klee}.
These results provide strong evidence that Formulog can be the basis of a realistic platform for SMT-based static analysis.

\subsection{Impact Beyond Formulog}

While we focus on Formulog, our work has broader impact along three main dimensions.

\paragraph{We push the frontiers of Datalog for static analysis}
Our work helps redefine the boundary of what is known to be feasible for Datalog-based static analysis, an active area of PL research that strives to make it possible to write analyses that are declarative (i.e., at the level of mathematical specifications) \emph{and} can be efficiently evaluated.
The original Formulog paper~\cite{formulog} primarily answers an expressivity-focused research question: ``Can SMT-based analyses like symbolic execution and refinement type checking be naturally expressed in a Datalog-like language?''
Our paper answers equally important performance-focused research questions: ``Can SMT-based analyses like symbolic execution and refinement type checking be efficiently evaluated in a Datalog-like language? Are existing ways of speeding up Datalog evaluation sufficient, or are new approaches necessary?''
In answering these questions---and by proposing a targeted evaluation algorithm that can be easily integrated with existing high-performance Datalog infrastructure---we give good evidence for the practicality of developing SAT/SMT-based analyses in Datalog, a research direction that is of interest to other (non-Formulog) declarative static analysis projects~\cite{saturn,Abeysinghe2024Flan,smaragdakis2021symbolic}.

\paragraph{We highlight how Datalog variants can be sensitive to logical inference order}
Other Datalog variants that interact with stateful external systems might benefit from evaluation modes, like eager evaluation, that explore the space of logical inferences in an order that is advantageous from the perspective of those external systems.
For example, a Datalog system that interacts with an external database that does not fit in memory would want to avoid inducing irregular database access patterns that lead to disk thrashing;
a Datalog system that interacts with AI chatbots (cf. Vieira~\cite{Li2024Relational}) might get better quality responses if the sequence of prompts sent to a chatbot forms a coherent line of conversation.

\paragraph{We question the conventional wisdom on speeding up Datalog}
Our experiences add some nuance to the conventional narrative of how to make non-distributed Datalog variants fast.
To this point, the suggestion has been that the combination of compilation, specialized data structures, and semi-naive evaluation is sufficient (perhaps with some additional attention to join orders and join algorithms). 
Our work demonstrates that this mindset runs the risk of missing out on performance improvements coming from the choice of evaluation algorithm itself, and that the ``frontend compilation'' approach of \citet{pacak2022functional}---while a reasonable starting point---can ultimately provide only limited performance gains if existing Datalog systems do not anticipate the potentially idiosyncratic workloads of the Datalog variant being compiled.
Over time, Datalog has moved wildly beyond its origins as a database query language, with applications in such diverse domains as static analysis~\citep{bddbddbUsing,bddbddbPointer,doop,datalogdisassembly,gigahorse,securify,madmax,smaragdakis2021symbolic,szabo2021incremental,inca}, networking~\citep{declarativenetworking,ryzhyk2019}, distributed systems~\citep{alvaro2010boom,alvaro2010dedalus}, access control~\citep{li2003datalog,dougherty2006specifying}, big data analytics~\citep{Shkapsky2016Big,Seo2013SociaLite}, and neurosymbolic AI~\citep{Huang2021Scallop,Li2023Scallop,Li2024Relational}.
As the applications of Datalog continue to evolve---whether to entirely new domains, or simply to new forms in existing domains (e.g., SMT-based static analyses)---we should keep in mind the possibility of developing novel Datalog evaluation algorithms to better fit new workloads.
A benefit of a declarative language like Datalog is that it is amenable to many different evaluation techniques; our work demonstrates one way that Datalog system designers can take advantage of this flexibility.

\subsection{Contributions}

At a high level, our paper explores how to efficiently execute SMT-based program analyses like symbolic execution and refinement type checking in a Datalog-like language.
The Formulog versions of these analyses differ from the types of analyses traditionally written in Datalog, in form (due to the extensive use of functional code), in the analysis logic they implement (thanks to SMT solving), and in their workloads (as SMT solving can be the most computationally intensive part).
By developing practical techniques to speed up these Datalog-based analyses, and showing that the Formulog versions can be competitive with---and sometimes much faster than---non-Datalog versions implemented in conventional languages, our work contributes to the larger research project of declarative, Datalog-based static analysis.
More concretely, we make the following contributions:
\begin{itemize}
  \item We design and evaluate a compiler from Formulog to off-the-shelf \souffle (Section~\ref{sec:compiler}).
  On the Formulog benchmark suite, the code produced by the compiler achieves an arithmetic mean 2.2$\times$ speedup over the Formulog interpreter's semi-naive evaluation mode.
  Through this compiler from Formulog to \souffle, we evaluate how well the recent ``frontend compilation'' proposal of \citet{pacak2022functional} works for Formulog.
  \item We observe that---unlike for traditional Datalog workloads---the \emph{order} in which logical inferences are made affects Formulog performance by determining the distribution of the SMT workload across threads.
        This inspires a novel ``eager'' strategy for parallel Datalog evaluation that eschews the batching of semi-naive evaluation in favor of eagerly pursuing the consequences of logical derivations with the help of a work-stealing thread pool (Section~\ref{sec:eagereval}).
  \item We extend the Formulog interpreter and \souffle's code generator to support eager evaluation, leading to arithmetic mean speedups of 5.2$\times$ and 7.6$\times$, respectively, over using off-the-shelf \souffle on SMT-heavy Formulog benchmarks; doing so demonstrates that eager evaluation is a practical and effective approach consistent with existing Datalog infrastructure (Section~\ref{sec:eagerimpl}).
  \item We reassess Formulog's performance in a wider context---showing that Formulog-based analyses can be competitive with previously published, hand-tuned, non-Datalog implementations of a range of SMT-based analyses---and suggest future improvements (Section~\ref{sec:discussion}).
\end{itemize}

\section{Background}\label{sec:background}

We overview the basics of Datalog (Section~\ref{sec:datalog}), \souffle (Section~\ref{sec:souffle}), and Formulog (Section~\ref{sec:formulog}). Throughout, we \textbf{boldface} a metavariable to indicate zero or more repetitions.

\subsection{Datalog}\label{sec:datalog}

\begin{figure}[t!]
  \begin{minipage}{1.0\linewidth}
    \sidebyside[0.63][0.36][t]{
      \[\begin{array}{@{}l@{}rcl@{}}
          \multicolumn{4}{l}{\textbf{Constructs}}                                                                    \\ \hline
          \text{Programs}     & \dl{\prog} & ::= & \many{\dl{H}}                                                     \\
          \text{Horn clauses} & \dl{H}     & ::= & p(\many{t}) \horn \many{\dl{A}}                                   \\
          \text{Atoms}        & \dl{A}     & ::= & p(\many{t}) \BNFALT !p(\many{t}) \BNFALT t = t \BNFALT t \noteq t \\
          \text{Terms}        & t          & ::= & X \BNFALT n \BNFALT @\cpp{f}(\many{t})
        \end{array}\]
    }{
      \[\begin{array}{@{}l@{}rcl@{}}
          \multicolumn{4}{l}{\textbf{Namespaces}}                  \\ \hline
          \text{Variables}     & X        & \in & \mathrm{Var}     \\
          \text{Integers}      & n        & \in & \mathbb{Z}       \\
          \text{Predicates}    & p        & \in & \mathrm{PredVar} \\
          \text{\cc functions} & ~\cpp{f} & \in & \mathrm{CppVar}
        \end{array}\]
    }
  \end{minipage}
  \caption{A grammar distilling the standard features of Datalog; we model \souffle by adding \emph{functor calls} $@\cpp{f}(\many{t})$, which are FFI calls invoking external \cc functions.}
  \label{fig:datalog}
\end{figure}

A Datalog~\citep{gallaire1978logic,neverdared,recursivequery} program is a set of Horn clauses (Figure~\ref{fig:datalog}), where each Horn clause is in the form $p(\many{t}) \horn \many{\dl{A}}$.
A Horn clause can be interpreted as a logical implication: the {\em head} predicate $p(\many{t})$ holds if all the {\em body} atoms $\many{\dl{A}}$ are true.
At a high level, Datalog evaluation amounts to making all logical inferences justified by the implications in the program, where each Horn clause is an inference rule in which the body atoms are the premises and the head predicate is the conclusion:
\[ p(t_1, \dots, t_m) \horn \dl{A}_1, \dots, \dl{A}_n \qquad \Longleftrightarrow \qquad \frac{\dl{A}_1 \quad \dots \quad \dl{A}_n}{p(t_1, \dots, t_m)} \]

An atom \dl{A} has one of four forms.
In addition to the standard \emph{positive predicates} $p(\many{t})$, Datalog typically supports \emph{negative predicates} $!p(\many{t})$.
A negative predicate holds when it cannot be derived by the program ({\em negation as failure}~\citep{clark1978negation}).
\emph{Unification predicates} $t = t$ unify their two arguments, failing if the unification is impossible.
\emph{Inequality predicates} $t \noteq t$ hold when their arguments are not equal.
All these predicates are defined over \emph{terms}, $t$, which are either variables $X$ or constants; here we assume, without loss of generality, that each constant is an integer $n$.
(We discuss \souffle's foreign function calls, $@\cpp{f}(\many{t})$, below.)

The predicate symbols in a Datalog program can be split into \emph{strata}, so that the predicate symbol of the head of a rule is not in a lower stratum than any predicate symbol occurring in its body.
By splitting a Datalog program into parts and putting a dependency order on them, stratification enables an intuitive and easily computable meaning for negated predicates in programs that meet the requirements of {\em stratified negation}~\citep{apt1988towards,przymusinski1988declarative,vangelder1989negation}.
The standard for Datalog, stratified negation sidesteps semantic and computational complications that arise when negation occurs within recursion (cf. the stable model semantics~\citep{stablemodel}), by requiring a predicate symbol $p$ to be in a strictly higher stratum than predicate symbol $q$ if there is a rule with predicate $p(\many{t})$ in the head and negated predicate $!q(\many{t}')$ in the body.

Datalog evaluation typically proceeds as a sequence of least-fixpoint computations, computing all the tuples in the predicates at the lowest stratum and working upwards, where each stratum is solved using some form of bottom-up evaluation.
For example, consider a stratum computing graph transitive closure, in which there is a non-recursive relation $\mathsf{edge}$ storing the edges of the graph (the input to the stratum) and a recursive relation $\mathsf{reach}$ defining their transitive closure:
\begin{lstlisting}
reach(X, Y) :- edge(X, Y).                       // rule trans_edge
reach(X, Z) :- edge(X, Y), reach(Y, Z).          // rule trans_step
\end{lstlisting}
{\em Naive evaluation} repeatedly runs these inference rules until no new $\mathsf{reach}$ tuples are derived (i.e., a fixpoint is reached).
Running \emph{all} the rules involves significant redundant computation---each iteration re-derives all tuples derived by previous iterations. The de facto standard technique of {\em semi-naive evaluation}~\citep{bancilhon1986naive} avoids some of this redundancy by executing Horn clauses rewritten to use auxiliary relations indexed by the iteration of the fixpoint computation:
\begin{lstlisting}
$\Delta$(reach)$^{[1]}$(X, Y) :- edge(X, Y).
$\Delta$(reach)$^{[i]}$(X, Z) :- edge(X, Y), $\delta$(reach)$^{[i-1]}$(Y, Z).
\end{lstlisting}
A predicate symbol $p^{[i]}$ signifies relation $p$ at the beginning of step $i$, and $p^{[i+1]} = \Delta(p)^{i} \cup p^{[i]}$ and $\delta(p)^{[i]} = \Delta(p)^{[i]} - p^{[i]}$.
Intuitively, the relation $\delta(p)^{[i]}$ is the set of new tuples derived in iteration $i$, and the relation $\Delta(p)^{[i]}$ is an overapproximation of that set.
We have reached a fixpoint for a stratum  once all of its $\delta$-relations are empty, i.e., we can learn nothing more.
Auxiliary relations are used only for predicates that are recursive in the current stratum; other relations are used directly.

Semi-naive evaluation is more complex in the presence of non-linearly recursive rules, where a single rule is rewritten into multiple semi-naive rules referring to different auxiliary relations.
Consider the nonlinearly recursive reformulation of rule \lstinline;trans_step;:
\begin{lstlisting}
reach(X, Z) :- reach(X, Y), reach(Y, Z).        // rule trans_guess
\end{lstlisting}
The Datalog engine rewrites \lstinline;trans_guess; into {\em two} rules, one per occurrence of a recursive predicate:
\begin{lstlisting}
$\Delta$(reach)$^{[i]}$(X, Z) :- $\delta$(reach)$^{[i-1]}$(X, Y), reach$^{[i]}$(Y, Z).
$\Delta$(reach)$^{[i]}$(X, Z) :- reach$^{[i]}$(X, Y), $\delta$(reach)$^{[i-1]}$(Y, Z).
\end{lstlisting}
Note that each rule accesses the \lstinline;reach; relation as it stood at the beginning of that iteration.%
\footnote{
  Technically, the first rule can have \icfoot{reach$^{[i-1]}$(Y, Z)} instead of \icfoot{reach$^{[i]}$(Y, Z)}, or alternatively the second rule can have \icfoot{reach$^{[i-1]}$(X, Y)} instead of \icfoot{reach$^{[i]}$(X, Y)}.
  This optimization materializes fewer combinations of tuples, since \icfoot{reach$^{[i-1]} \subseteq$ reach$^{[i]}$} (by the monotonicity of Datalog).
  However, neither \souffle nor Formulog uses it.
}
Readers looking to learn more can consult Datalog surveys~\citep{neverdared,recursivequery}.

\subsection{\souffle}\label{sec:souffle}

\souffle~\citep{souffle1,souffle2} is a leading Datalog implementation designed specifically for static analysis workloads.
It has many features; we focus on those relevant to this paper.
Foremost among these, \souffle is itself a compiler: it compiles Datalog programs to an intermediate representation---instructions on a notional relational algebra machine (RAM)---and then compiles these instructions to \cc code.
The resulting code can be linked against external \cc functions, which can be invoked in \souffle source code as \emph{functor calls}, written $@\cpp{f}(\many{t})$ (Figure~\ref{fig:datalog}).

A Horn clause is typically compiled into nested \lstinline;for; loops.
For example, the semi-naive version of the \lstinline;trans_step; rule (from above) is compiled into these loops (given in pseudo-Python):
\begin{lstlisting}
pfor (x, y) in edge:
    lb = (y, DOMAIN_MIN) # key for lower bound
    ub = (y, DOMAIN_MAX) # key for upper bound
    for (_, z) in $\delta$(reach)$^{[i-1]}$.range_inclusive(lb, ub):
        if (x, z) not in reach$^{[i]}$:
            $\delta$(reach)$^{[i]}$.insert((x, z))
\end{lstlisting}
\souffle parallelizes the outermost loop using an OpenMP parallel \lstinline;for; loop~\citep{dagum1998openmp}.
This loop iterates over the entire \lstinline;edge; relation; the next loop iterates over a slice of the $\delta(\mathsf{reach})^{[i-1]}$ relation that is extracted by making a range query (here, over the entire domain).
Much of the engineering and research behind \souffle has gone into making these data structure operations efficient.
First off, \souffle uses an optimal indexing strategy minimizing the number of indices that need to be maintained for a relation~\citep{souffleindex}.
Second, \souffle uses concurrent tries~\citep{souffletrie} and B-trees~\citep{soufflebtree} that have been specialized for Datalog evaluation.
These optimizations help make \souffle a top-performing Datalog engine on static analysis workloads~\citep{recstep}, as evidenced by its use as the backbone of a slew of static analysis platforms~\cite{gigahorse,madmax,datalogdisassembly,smaragdakis2021symbolic,Antoniadis2017Porting}.

\subsection{Formulog}\label{sec:formulog}

\begin{figure}[t]
  \[ \begin{array}{@{}l@{}rcl@{}}
      \multicolumn{4}{l}{\textbf{Constructs}}                                                                                                                                \\ \hline

      \text{Programs}     & \flg{\prog} & ::= & \many{T}~\many{F}~\many{\flg{H}}                                                                                             \\
      \text{Functions}    & F           & ::= & \mathsf{fun}~\flg{f}(\many{X}) = e                                                                                           \\
      \text{Horn clauses} & \flg{H}     & ::= & p(\many{e}) \horn \many{\flg{A}}                                                                                             \\
      \text{Atoms}        & \flg{A}     & ::= & p(\many{e}) \BNFALT !p(\many{e}) \BNFALT  e == e \BNFALT e \noteq e \BNFALT X \leftarrow e \BNFALT e \rightarrow c(\many{X}) \\
      \text{Constants}    & k           & ::= & \mathsf{true} \BNFALT \mathsf{false} \BNFALT n \BNFALT \dots                                                                 \\
      \text{Expressions}  & e           & ::= & X \BNFALT k
      \BNFALT c(\many{e})
      \BNFALT \flg{f}(\many{e})
      \BNFALT p(\many{e})
      \BNFALT \letin{X}{e}{e} \BNFALT                                                                                                                                        \\
                          &             &     & \ite{e}{e}{e} \BNFALT \match{e}{[c(\many{X}) \rightarrow e]^*}
    \end{array}\]

  \[\begin{array}{lrclclrcl}
      \multicolumn{9}{l}{\textbf{Additional namespaces}}                                        \\ \hline
      \text{Type definitions} & T                   & \in & \mathrm{TypeVar} &
      \quad                   & \text{Constructors} & c   & \in              & \mathrm{CtorVar} \\
      \text{Function symbols} & \flg{f}             & \in & \mathrm{FunVar}
    \end{array} \]
  \caption{A grammar for core Formulog, which extends Datalog with functional programming via algebraic data type definitions $T$, functions $F$, and expressions $e$.
    The functional language includes constructors for building terms representing SMT formulas and functions for testing their satisfiability and extracting models.}
  \label{fig:formulog}
\end{figure}

Formulog~\citep{formulog} is a recent Datalog variant for implementing static analyses that use satisfiability-modulo-theories (SMT) solving~\citep{smt_survey}, a key automated reasoning technology that enriches boolean satisfiability (SAT) solving with predicate logic and the ability to reason about common program constructs such as arrays, machine integers, and IEEE floating point numbers.

We focus on a minimal core representation of Formulog (Figure~\ref{fig:formulog}).
The key characteristic of this representation is that unification predicates $e = e$ have been replaced with three new types of atoms: check atoms, assignment atoms, and destructor atoms.
This representation is simpler to interpret and compile, as it separates out and sequences heterogeneous operations that can occur implicitly within a single unification.
A \emph{check atom} $e == e$ holds if two expressions evaluate to the same value, like an equality predicate.
An \emph{assignment atom} $X \leftarrow e$ assigns the value of an expression $e$ to a variable $X$ (which cannot be bound to a value beforehand; atoms are evaluated left-to-right).
A \emph{destructor atom} $e \rightarrow c(\many{X})$ fails if $e$ does not evaluate to a value with outermost constructor $c$; otherwise, it assigns the constructor's arguments to the variables \many{X} (which cannot be bound).
We write $e \noteq e$ to negate a check atom---like an inequality predicate---but assignment and destructor atoms cannot be negated.

Beyond Horn clauses, Formulog supports first-order functional programming and SMT solving.

Formulog embeds a functional language with standard features like algebraic data types and pattern matching.
It also provides the ability for functional code to reference Datalog relations by treating a predicate symbol $p$ as a function symbol.
For example, if the tuple \lstinline;(0, 1); is a member of the relation \lstinline;reach;, the expression \lstinline;if reach(0, 1) then "yes" else "no"; evaluates to the string \lstinline;"yes";, and evaluates to \lstinline;"no"; otherwise.
Formulog requires the use of predicates-as-functions to be stratified: if evaluating a rule with head predicate symbol $p$ might run functional code that invokes predicate $q(\many{e})$, predicate symbol $q$ must be in a lower stratum than predicate symbol $p$.

Formulog's SMT support is built on top of its functional programming support: the various SMT term constructors are simply built-in Formulog-level constructors.
The built-in function \lstinline;is_sat; serializes its argument---a complex term representing an SMT formula---into the SMT-LIB standard format~\citep{smtlib} and discharges it to an external SMT solver to check for satisfiability.
Formulog's SMT API makes it easy to check validity and extract models, too.

The prototype Formulog system is 20k SLOC Java and uses semi-naive evaluation by default.%
\footnote{All SLOC counts given in this paper are of non-blank, non-comment, physical lines.}
The Formulog runtime supports parallel evaluation, and each evaluation worker thread interacts with its own stateful, external solver process.
Modern SMT solvers support incremental solving~\citep{Een2003Temporal}, which benefits from query {\em locality}: solving is (often) faster if consecutive queries to a given solver instance are similar.
Formulog uses a few encoding tricks to increase the number of conjuncts shared between queries on a given thread~\citep{formulog_incremental_smt}.

\section{Compiling Formulog to \souffle}\label{sec:compiler}

This section develops a compiler from Formulog to off-the-shelf \souffle, first presenting the high-level strategy (Section~\ref{sec:core_transl}) and then implementation details (Section~\ref{sec:transl_details}).
In the mean, the generated code has a 2.2$\times$ speedup over the Formulog interpreter and uses 38$\times$ less memory (Section~\ref{sec:evalcompiler}).

\subsection{Core Translation}\label{sec:core_transl}

The basic compilation strategy is to translate Formulog's Datalog fragment to \souffle code and its functional fragment to \cc code accessible to the \souffle code via external functors (Figure~\ref{fig:translation}).

\begin{figure}[p]

  \hdr{Contexts and additional namespaces}{}
  \vspace{-10pt}
  \[ \begin{array}{lrcl}
      \text{Function and code contexts} & \Gamma & ::=       &
      \cdot \BNFALT
      \Gamma, \flg{f} \mapsto \cpp{f} \BNFALT
      \Gamma, C
      \\

      \text{Logic variable contexts}    & \Delta & \subseteq & \mathrm{Var}
      \\%[0.5em]

      \text{\cc code}                   & C      & \in       & \mathrm{CppCode} \\
    \end{array} \]

  \hdr{Program translation}{
    \quad \fbox{$\vdash \transl{\flg{\prog}} \Rightarrow \many{\dl{H}}, C$}
  }

  \infrule[Prog]{
  \Gamma_0 = makeFlgRuntime(\many{T})
  \andalso
  \dl{\many{H}}_{admin} = makeAdminRules(\many{F}~\many{\flg{H}}) \\
  |\many{F}| = m
  \andalso
  \forall i \in [0, m).~\Gamma_i \vdash \transl{\many{F}[i]} \Rightarrow \Gamma_{i + 1} \\
  |\many{\flg{H}}| = n
  \andalso
  \forall i \in [0, n).~\Gamma_{m + i} \vdash \transl{\many{\flg{H}}[i]} \Rightarrow \dl{H}_i, \Gamma_{m + i + 1}
  \andalso
  C = extractCode(\Gamma_{m + n})
  }{
  \vdash \transl{\many{T}~\many{F}~\many{\flg{H}}} \Rightarrow \dl{\many{H}}_{admin} \concat [\dl{H}_0, \dots, \dl{H}_{n-1}], C
  }

  (helper functions are defined in \iffull Appendix~\ref{sec:helper_funcs}\else \ifanonymous the supplemental material\else the extended version of this paper~\citep{full}\fi\fi)

  \hdr{Function translation}{
    \quad \fbox{$\Gamma \vdash \transl{F} \Rightarrow \Gamma$}
  }

  (definition omitted; follows standard techniques)

  \hdr{Clause translation}{
    \quad \fbox{$\Gamma \vdash \transl{\flg{H}} \Rightarrow \dl{H}, \Gamma$}
  }

  \infrule[Clause]{
  |\flg{\many{A}}| = n
  \andalso
  \Delta_0 = vars(p(\many{e}) \horn \many{\flg{A}})
  \\
  \forall i \in [0, n).~\Gamma_i, \Delta_i \vdash \transl{\many{\flg{A}}[i]} \Rightarrow \dl{\many{A}}_i,\Gamma_{i+1}, \Delta_{i+1}
  \\
  \Gamma_n, \Delta_n \vdash \transl{p(\many{e})} \Rightarrow [p(\many{t})], \Gamma_{final}, \_
  \andalso
  \dl{\many{A}}_{body} = \dl{\many{A}}_0 \concat \cdots \concat \dl{\many{A}}_{n-1}
  }{
  \Gamma_0 \vdash \transl{p(\many{e}) \horn \many{\flg{A}}} \Rightarrow p(\many{t}) \horn \dl{\many{A}}_{body}, \Gamma_{final}
  }

  \hdr{Atom translation}{
    \quad \fbox{$\Gamma, \Delta \vdash \transl{\flg{A}} \Rightarrow \many{\dl{A}}, \Gamma, \Delta$}
  }

  \infrule[PosPred]{
  |\many{e}| = n
  \andalso
  \forall i \in [0,n).~\Gamma_i \vdash \transl{\many{e}[i]} \Rightarrow t_i, \Gamma_{i + 1}
  }{
  \Gamma_0, \Delta \vdash \transl{p(\many{e})} \Rightarrow [p(t_0, \dots, t_{n-1})], \Gamma_n, \Delta
  }

  \sidebyside[0.5][0.5]{
    \infrule[Check]{
      \Gamma_0 \vdash \transl{e_1} \Rightarrow t_1, \Gamma_1
      \andalso
      \Gamma_1 \vdash \transl{e_2} \Rightarrow t_2, \Gamma_2
    }{
      \Gamma_0, \Delta \vdash \transl{e_1 == e_2} \Rightarrow [t_1 = t_2], \Gamma_2, \Delta
    }
  }
  {
    \infrule[Assign]{
      \Gamma \vdash \transl{e} \Rightarrow t, \Gamma'
    }{
      \Gamma, \Delta \vdash \transl{X \leftarrow e} \Rightarrow [X = t], \Gamma', \Delta
    }
  }

  \infrule[Destruct]{
  \Gamma \vdash \transl{e} \Rightarrow t, \Gamma'
  \andalso
  \{X,Y,Z\} \cap \Delta = \emptyset
  \\
  \dl{\many{A}}_{dtor} = [
  X = t,
  @\mathsf{is\_ctor}_c(X) = Y,
  \mathsf{move\_barrier}(Y, Z)
  ]
  \\
  \dl{\many{A}}_{assign} = [
  X_0 = @\mathsf{nth}(0, X, Z),
  \dots,
  X_{n-1} = @\mathsf{nth}(n-1, X, Z)
  ]
  }{
  \Gamma, \Delta \vdash \transl{e \rightarrow c(X_0, \dots, X_{n-1})} \Rightarrow \dl{\many{A}}_{dtor} \concat \dl{\many{A}}_{assign}, \Gamma', \Delta \cup \{X, Y, Z\}
  }

  \hdr{Expression translation}{
    \quad \fbox{$\Gamma \vdash \transl{e} \Rightarrow t, \Gamma$}
  }
  \vspace{10pt}

  \sidebyside[0.3][0.65][c]{
    \infrule[Var]{}{
      \Gamma \vdash \transl{X} \Rightarrow X, \Gamma
    }
  }
  {
    \infrule[NonVar]{
      e \not \in \mathrm{Var}
      \andalso
      \flg{f}~\text{fresh in}~\Gamma
      \andalso
      \many{X} = freeVars(e)
      \\
      \Gamma \vdash \transl{\mathsf{fun}~\flg{f}(\many{X}) = e} \Rightarrow \Gamma'
      \andalso
      \Gamma'(\flg{f}) = \cpp{f}
    }{
      \Gamma \vdash \transl{e} \Rightarrow @\cpp{f}(\many{X}), \Gamma'
    }
  }

  \caption{Rules for translating from Formulog to \souffle.
    Negative predicates $!p(\many{e})$ and inequality predicates $e \noteq e$ can be translated analogously to their positive counterparts.
    We use the notation $\concat$ for list concatenation.
  }
  \label{fig:translation}
\end{figure}

At the top level (\rn{Prog}), the translation takes as input a Formulog program consisting of type definitions, function definitions, and Horn clauses, and produces as output (\souffle) Horn clauses and \cc code.
Translating a program consists of four main parts.
First, the Formulog type information is used to instantiate a skeleton \cc Formulog runtime.
The runtime includes the standard library and interface to external SMT solvers; type information is necessary to complete the representation of terms and their serialization into SMT-LIB.
Second, the Horn clauses and function definitions are examined to produce administrative rules that constrain the \souffle code to respect the semantics of the source Formulog program, such as having the same stratification.
Third, the function definitions are compiled to \cc functions using standard compilation techniques for functional languages (the only nonstandard feature being predicates $p(\many{e})$ invoked as functions; discussed in Section~\ref{sec:reify_relations}).
Finally, each Horn clause in the Formulog program is translated in turn (\rn{Clause}).

Translating a clause consists of translating the atoms in order, starting with the leftmost body atom and concluding with the head of the rule (\rn{Clause}).
During this process, the compiler keeps track of the variables that are used in the rule, so that it can generate fresh variables on demand.
Generic predicates $p(\many{e})$ (\rn{PosPred}) and $!p(\many{e})$ (rule omitted), equality checks $e == e$ (\rn{Check}) and $e \noteq e$ (rule omitted), and assignment $X \leftarrow e$ (\rn{Assign}) all compile down to single \souffle atoms.
On the other hand, destructors $e \rightarrow c(\many{X})$ compile down to {\em multiple} \souffle atoms (\rn{Destruct}), due to a limitation in the current version of \souffle (Section~\ref{sec:compiling_dtors_main}).

Expressions, which occur as arguments to atoms, are compiled to \souffle terms two ways.
Variables are translated using the identity function (\rn{Var}).
For all other expressions (\rn{NonVar}), the compiler creates a fresh Formulog function definition with that expression as its body and the expression's free variables as its arguments, and translates the function to \cc; the compiler then emits an external functor call invoking the \cc function corresponding to the expression.

\subsection{Implementation Details}\label{sec:transl_details}

Our compiler is 4,000 SLOC Java and the runtime is 3,000 SLOC \cc.

\subsubsection{Representing Values}

Formulog values are represented as integers in the \souffle code and as \cc objects in our \cc runtime.
The runtime performs hash consing during all value creation, so that there is exactly one \cc object per Formulog value.
This allows us to use the memory address of a Formulog object as the integer representation of that value in \souffle; translating values as they flow between \souffle and our runtime is as simple as casting.

An alternative would be to embed Formulog values into \souffle's type system, which supports algebraic data types.
We decided against this approach for four reasons.
First, there is not a natural translation between Formulog's type system and \souffle's.
For example, Formulog has multiple widths of integers, whereas \souffle has just one; Formulog data types have ML-style parametric polymorphism, whereas \souffle has templates.
Second, our approach gives us more control over the memoization of terms.
Third, it frees us from having to use \souffle's internal representation of complex values in the \cc code we generate for Formulog expressions.
Fourth, \souffle (as of v2.4) does not support unpacking complex terms returned from external functors---i.e., getting their arguments---and so most manipulation of terms would need to occur outside \souffle anyway.

\subsubsection{Compiling Destructors}\label{sec:compiling_dtors_main}

Because of a current limitation in \souffle, each Formulog destructor atom is compiled to {\em multiple} \souffle atoms (\rn{Destruct}; see discussion in \iffull Appendix~\ref{sec:compiling_dtors}\else \ifanonymous the supplemental material\else the extended version of this paper~\citep{full}\fi\fi).
It should in principle be straightforward to compile a destructor atom to a single \souffle atom using \souffle's support for complex terms; however, this currently leads to an assertion failure within the \souffle code generator.
While we expect this limitation to be fixed in a future version of \souffle, we do not believe that the current clunky translation leads to a substantial runtime cost.

\subsubsection{Reifying Relations}\label{sec:reify_relations}

In addition to standard functional programming fare, Formulog supports querying relations from functional code, by invoking predicates $p(\many{e})$ as functions.
We need to account for this both in our generated \cc code and \souffle code.
On the \cc side, when the main driver of the executable creates an instance of the \souffle program (i.e., a value of type \lstinline;souffle::SouffleProgram;), it stores it into a global variable before executing it.
This allows us to generate \cc code that looks up the content of relations while the \souffle program is running.
On the \souffle side, we create administrative rules to make sure that the stratification of the \souffle program lines up with the stratification of the Formulog program: the use of predicates as functions must be stratified (as mentioned in Section~\ref{sec:formulog}), but \souffle does not see the functional code.
For example, a rule defining a nullary relation \lstinline;p; might invoke a function \lstinline;f;, which in turn queries the emptiness of a nullary relation \lstinline;q;.
To force \souffle to put relation \lstinline;p; in a higher stratum than \lstinline;q;, we generate the rule \lstinline;p() :- empty(), !q().; (where \lstinline;empty; is an empty administrative relation).

\subsection{Evaluation}\label{sec:evalcompiler}

\begin{figure}[t]
  \includegraphics{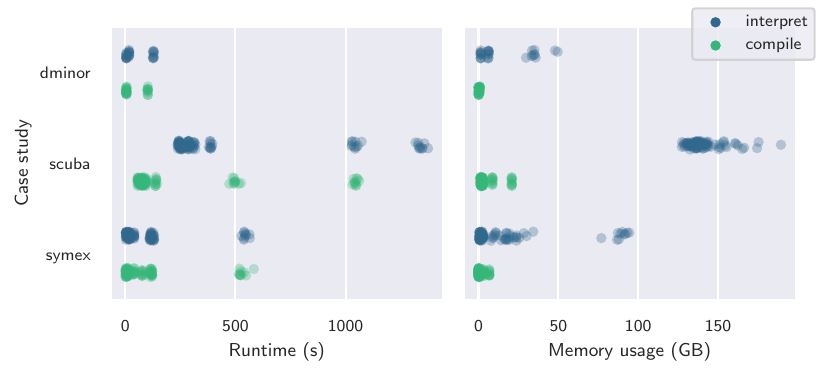}
  \caption{
    Formulog compiled to off-the-shelf \souffle achieves a 2.2$\times$ mean speedup over the baseline Formulog interpreter, while also using 38$\times$ less memory in the mean.
    Left is better; we omit compilation time (static analyses are compiled once and executed many times).
  }
  \label{fig:init_compiler}
\end{figure}

The code generated by our compiler beats the Formulog interpreter on the case studies and benchmarks from the Formulog paper~\citep{formulog} (Figure~\ref{fig:init_compiler}), achieving a mean speedup of 2.2$\times$ (min\sep{}median\sep{}max: \speedups{1.0}{1.7}{4.3}) and using 38$\times$ less memory in the mean (min\sep{}median\sep{}max: \speedups{6.7}{20}{110}).%
\footnote{
  All means reported in this paper are arithmetic.
  \ifanonymous
    See the supplementary material for raw results and data analysis scripts.\fi}
The baseline interpreter we compared against is an improved version of the one used in the original Formulog paper (e.g., it now completes on all benchmarks, whereas it previously timed out on three); we also manually tuned some of the baseline interpreter's parameters to achieve better performance, and resolved various performance anomalies.

We ran all experiments on an Ubuntu 22.04 m5.12xlarge Amazon Web Services EC2 instance, with an Intel Xeon processor clocked at \SI{3.1}{GHz}, 24 physical cores (48 vCPUs), and \SI{192}{GiB} RAM.
We use OpenJDK v17.0.7 as our Java runtime, g\texttt{++} v11.3.0 as our \cc compiler, Z3~\citep{z3} v4.12.2 (commit f7c9c9e) as our SMT solver, and \souffle v2.4.
Both the generated code and the Formulog interpreter are configured to use 40 threads.

To compute the runtime on a benchmark, we take the median of 10 trials.
We omit compilation times, on the basis that the Formulog programs in this benchmark suite are static analyses that are intended to be compiled once and then executed many times.
Compilation time is dominated by the time it takes to compile the \cc code generated by \souffle, which can take up to a few minutes.

The benchmarks fall into three case studies representing different SMT-based static analyses, each consisting of 1-1.5K lines of Formulog code (for details, see the original Formulog paper~\citep{formulog}).
We updated the code to work with newer versions of Formulog and to avoid unevaluated expressions in input facts (which the \cc runtime parser does not currently support).

\paragraph{Refinement Type Checking (\bm{dminor})}
This case study is a type checker for Dminor~\citep{dminor}, a language combining refinement types and dynamic type tests.
We categorize this case study as ``SMT heavy'' as the SMT solver is frequently invoked during type checking to prove subtyping relations and the termination of expressions (in the longest-running benchmark, 96k external SMT solver calls are made).
The SMT formulas are also particularly complex, as they touch on many SMT theories and include constructs like universally quantified formulas.
The three benchmarks are synthetic, created by composing all publicly available Dminor programs (that use the core language only) and then copying the composite program $n$ times, for $n \in \{1, 10, 100\}$.%
\footnote{The copies are transformed so that they are not syntactically equal, which in turn means that the SMT queries generated by Formulog for each copy are also not syntactically equal.
Because the base program exercises many parts of the core Dminor calculus, we believe this approach of~\citet{formulog} to synthesizing benchmarks leads to a reasonable approximation of the workloads that would be induced by variously sized Dminor developments, were they to exist.}
For the benchmark with 100 copies of the base program (\bm{all-100}), there are 4600 input tuples and 1.7 million output tuples; the most computationally intensive part of this case study is the SMT solving.
On this case study the compiler achieved a mean speedup of 1.9$\times$ (min\sep{}median\sep{}max: \speedups{1.2}{1.2}{3.2}).

\paragraph{Bottom-Up Points-To Analysis (\bm{scuba})}
This context-sensitive, bottom-up points-to analysis for Java uses SMT formulas to summarize methods~\citep{bottomupPointsto}.
While SMT solving is used when function summaries are instantiated, we categorize this case study as “SMT light” because few external SMT calls are made in practice (many potential SMT queries can instead be resolved by basic preprocessing done in the analysis code).
Additionally, a relatively small proportion of the analysis (20\%) is written in Formulog's functional fragment; hence, of the three case studies, this one most closely resembles conventional Datalog programs.
The 10 benchmarks are substantial Java applications used in the evaluation of the original implementation of this analysis~\citep{bottomupPointsto}, drawn primarily from existing Java benchmark suites including DaCapo~\citep{dacapo}.
Accordingly, these benchmarks have large numbers of input and output facts, with input fact databases ranging from 65k tuples (weblech) to 1.5 million tuples (xalan), and output fact databases ranging from 4.3 million tuples (polyglot) to 120 million tuples (sunflow).
On this case study the compiler achieved a mean speedup of 3.2$\times$ (min\sep{}median\sep{}max: \speedups{1.3}{3.4}{4.3}).

\paragraph{Symbolic Execution (\bm{symex})}
This case study is a bounded symbolic executor for a fragment of LLVM bitcode in the style of KLEE~\citep{klee}.
The 10 benchmark programs are compiled C programs that perform operations like array sorting, solving logic puzzles, and interpreting bitcode.
As the benchmarks are synthetic and need to fit within the supported fragment of LLVM, the source code for each benchmark is generally small in terms of lines of code; however, because many of the programs loop over symbolic data, the number of states to symbolically explore can be quite large.
We categorize this case study as “SMT heavy” as SMT calls are made frequently, any time the symbolic executor reaches a branch point conditioned on symbolic data (220k external SMT solver calls are made in sort-7, the benchmark with the most SMT calls).
The largest input fact database is 6000 tuples (prioqueue-6), but the largest output fact database is 28 million tuples (sort-7).
On this case study the compiler achieved a mean speedup of 1.4$\times$ (min\sep{}median\sep{}max: \speedups{1.0}{1.4}{2.0}).

\section{Eager Evaluation}\label{sec:eagereval}

The \souffle code generated by our compiler outperforms the Formulog interpreter when both perform semi-naive evaluation.
This is hardly surprising, as using \souffle has many advantages over using the Formulog prototype, such as compilation instead of interpretation, \cc instead of Java, and data structures specialized for Datalog instead of generic ones.

What {\em is} surprising is that we can do better, and with relatively little effort.
There are many possible ways to evaluate a Datalog program, which is declarative: a Datalog program comprises some logical inference rules with no stipulations as to how to discover the rules' consequences. We need not be beholden to the standard strategy.
In fact, one might even expect that new Datalog variants would put different stresses on Datalog engines than traditional Datalog workloads, and thus would benefit from alternative Datalog evaluation strategies.
(In fact, even different \emph{Datalog} workloads benefit from different implementation techniques~\citep{recstep}.)

Formulog indeed differs from traditional Datalog: external SMT solving is typically one of the most expensive parts of evaluating a Formulog program.
This section presents \emph{eager evaluation}, an alternative, worklist-based strategy for parallel Datalog evaluation that achieves a quasi-depth-first search of the logical inference space by submitting inference tasks to a work-stealing thread pool.%
\footnote{In a work-stealing thread pool~\citep{mohr1990lazy,arora1998thread,blumofe1999scheduling}, each worker thread maintains a last-in-first-out (LIFO) stack of work items.
  Since work items in our setting correspond to logical derivations, the LIFO discipline leads to the most recent derivations being explored first; see additional discussion in Section~\ref{sec:eager_description}.
}
Compared to semi-naive evaluation, which performs a breadth-first search, eager evaluation entails a different distribution of SMT calls across worker threads (and thus external SMT solvers); in practice, it tends to lead to better SMT solving times.
For example, on some Formulog programs, eager evaluation ends up putting similar SMT calls on the same worker thread, leading to more effective incremental SMT solving.
Using eager evaluation, the Formulog interpreter beats the semi-naive code generated by \souffle on SMT-heavy benchmarks (Section~\ref{sec:eagerevalresults}).

This section motivates eager evaluation (Section~\ref{sec:eager_motivation}), describes the high-level algorithm (Section~\ref{sec:eager_description}), and proves its correctness (Section~\ref{sec:eager_correctness}).

\subsection{Motivating Eager Evaluation}\label{sec:eager_motivation}

\begin{figure}[t!]
  \centering
  \begin{minipage}{.5\textwidth}
    \centering
    \includegraphics[width=0.75\textwidth]{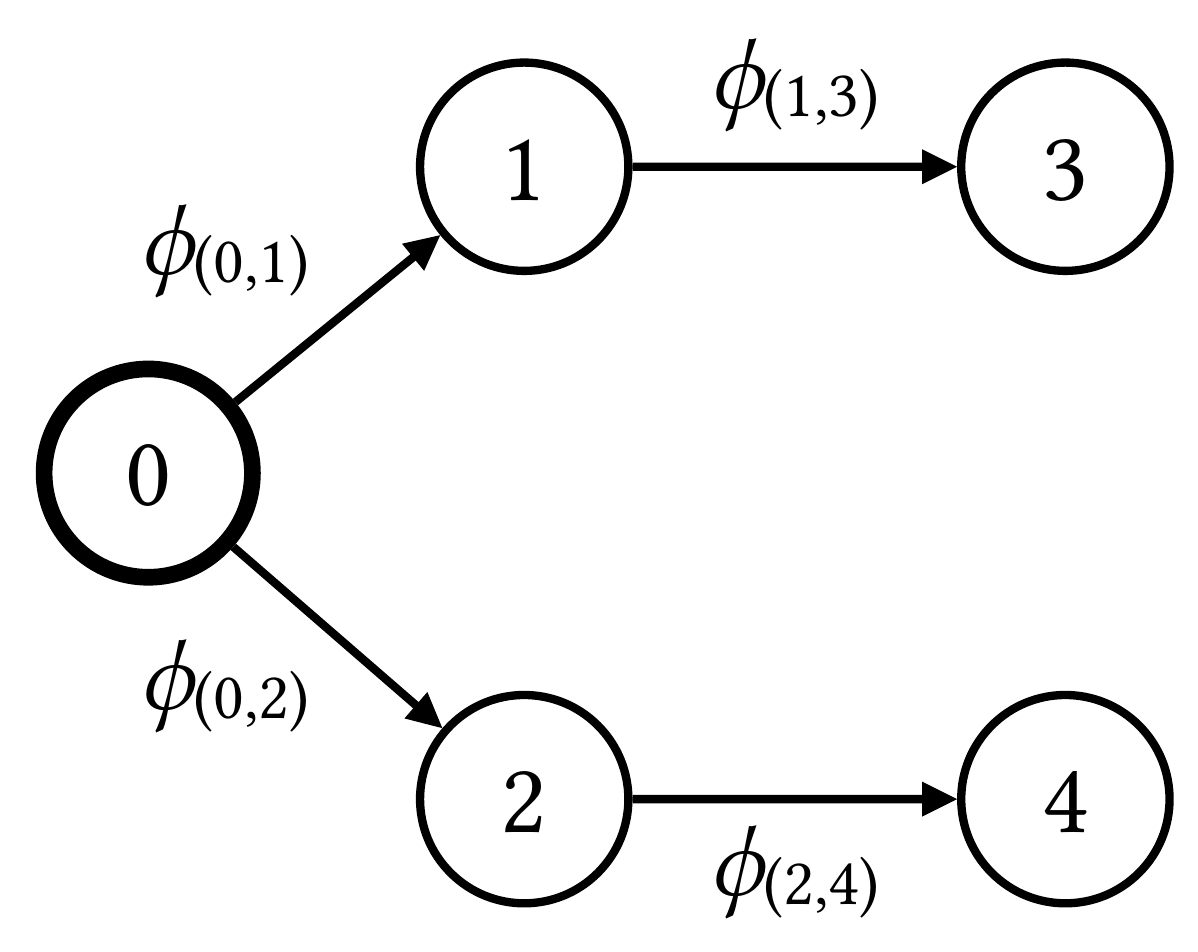}
  \end{minipage}%
  \begin{minipage}{.5\textwidth}
    \centering
    \begin{lstlisting}
(* base case *)
reach(0, $\bk$true$\bk$). 

(* recursive case *)
reach(Y, $\bk$Phi /\ Psi$\bk$) :-
  reach(X, Phi),
  edge(X, Y, Psi),
  is_sat($\bk$Phi /\ Psi$\bk$) = true.
      \end{lstlisting}
  \end{minipage}
  \caption{A Formulog program computes reachability over a tree labeled with SMT propositions $\phi_{(i,j)}$ by using the built-in constructor for conjunction (\lstinline;/$\texttt{\textbackslash}$;) and the built-in function \lstinline;is_sat; to check satisfiability.
    Backticks demarcate SMT formulas.
  }
  \label{fig:reachability}
\end{figure}

To motivate eager evaluation, consider the task of computing reachability on trees where the edges are labeled with logical propositions.
The root node is trivially reachable.
Some other node $n$ is reachable if its parent $m$ is reachable from the root node via a $k$-hop path with edge labels $\phi_1, \dots, \phi_k$, there is an outgoing edge from node $m$ with label $\psi$ to a child node $n$, and the conjunction $\phi_1 \wedge \cdots \wedge \phi_k \wedge \psi$ is satisfiable---that is, reachability is as in a directed graph, modulo satisfiability of path conditions.
This setting exemplifies a fundamental problem in real static analyses, mimicking how symbolic execution explores program execution trees.

For the sake of concreteness, say we run the Formulog program computing this type of graph reachability on a sample tree (Figure~\ref{fig:reachability}), using a single thread of execution.
If we use semi-naive evaluation, the program will discover nodes in a breadth-first order: 0-1-2-3-4.
Doing so involves four SMT calls, which occur in this order: $\phi_{(0,1)}$, $\phi_{(0,2)}$, $\phi_{(0,1)} \wedge \phi_{(1,3)}$, and $\phi_{(0,2)} \wedge \phi_{(2,4)}$.
Because adjacent calls do not share any conjuncts, the SMT solver is not able to naturally perform incremental SMT solving; in the worst case, it would have to clear its state between calls that \emph{do} share conjuncts.
And, while Formulog does use some tricks to encode conjuncts in the SMT solver's state so that they can be disabled and enabled in subsequent calls~\citep{formulog_incremental_smt}, these techniques are necessarily limited and introduce overhead.

In contrast, if one were to compute reachability over the example graph using a Datalog evaluation that simulates depth-first search, we might encounter the nodes in the order 0-1-3-2-4, inducing SMT calls in the order $\phi_{(0,1)}$, $\phi_{(0,1)} \wedge \phi_{(1,3)}$, $\phi_{(0,2)}$, and $\phi_{(0,2)} \wedge \phi_{(2,4)}$.
This order naturally leads to incremental SMT solving, since the first two calls share a conjunct, as well as the last two calls.
For example, by checking the satisfiability of the proposition $\phi_{(0,1)}$, the solver might learn information useful for checking the satisfiability of the conjunction $\phi_{(0,1)} \wedge \phi_{(1,3)}$.

In general, semi-naive evaluation computes a breadth-first search through the logical inference space: it computes derivations of height $k$ only after it has computed all derivations of height less than $k$.
For the example of graph reachability, this quite literally matches breadth-first graph traversal.
If we were in the world of Datalog and concerned solely with computing structural reachability in the graph (irrespective of the satisfiability of path conditions), there would not be an obvious cost to performing breadth-first search compared to some other traversal.
However, given that we are in the world of Formulog and need to respect the satisfiability of path conditions, order \emph{does} matter: depth-first search leads to more opportunities for incremental SMT solving.
Put another way, we want a Datalog evaluation strategy that prioritizes computing the consequences of the most recently derived fact---i.e., finding those facts where the most recently derived fact appears in the proof tree---instead of processing facts in order of the height of their proof trees.

This is what motivates eager evaluation: a Datalog evaluation strategy that, by being closer to depth-first search, leads to different and hopefully advantageous SMT workloads for Formulog programs, while staying amenable to effective parallelization.

\subsection{Eager Evaluation: Parallel, Most-Recent-First Search}\label{sec:eager_description}

Eager evaluation explores the inference space using a parallelized, DFS-like search.
Making derivations in a strict depth-first order would lead to poor parallelism; instead, eager evaluation performs a more relaxed traversal strategy we think of as ``most-recent-first search.''
This retains the spirit of DFS---it prioritizes processing the most recently derived facts---while being naturally parallelizable.

Similar to semi-naive evaluation, eager evaluation is applied one stratum at time, making it compatible with stratified negation~\citep{apt1988towards,przymusinski1988declarative,vangelder1989negation} and stratified aggregation~\citep{mumick1990magic}.
It is a worklist-based scheme: initial work items are submitted to the worklist; work items recursively submit new work items; evaluation is finished once the worklist is empty.
To both parallelize evaluation and achieve a most-recent-first order, eager evaluation uses a work-stealing thread pool for the worklist~\citep{mohr1990lazy,arora1998thread,blumofe1999scheduling}.
Each worker thread in the pool maintains a dequeue of work items that it pops from the back using the last-in-first-out stack discipline; since a work item corresponds to a new inference to explore, this discipline prioritizes recent derivations, and thus approximates a depth-first traversal through the logical inference space.
If a worker thread's dequeue becomes empty, it steals from the front of a random thread's dequeue using the first-in-first-out queue discipline, leading to lower thread contention, and often good locality.

Each work item submitted to the thread pool represents a rule to evaluate.
There are two types of work items, representing different types of rules.
The first type represents a rule that is not recursive in the current stratum; that is, all the body predicates are defined in previous strata.
These are the initial work items (one per non-recursive rule).
The second type of work item represents a recursive rule that has been specialized to a particular fact $p(\many{n})$.
To specialize a rule, 1) unify a body atom $p(\many{t})$ with the fact $p(\many{n})$ to get a substitution from variables to constants; 2) apply this substitution across the rule; and 3) reorder the rule so that the body atom $p(\many{n})$ is first.
For example, consider the \lstinline;trans_step; rule for graph transitive closure (Section~\ref{sec:datalog}):
$$\text{\lstinline;reach(X, Z) :- edge(X, Y), reach(Y, Z).;}$$
Specialized to the fact \lstinline;reach(0, 1);, the rule would become
$$\text{\lstinline;reach(X, 1) :- reach(0, 1), edge(X, 0).;}$$
Specialization can fail if the body atom $p(\many{t})$ cannot be unified with the fact $p(\many{n})$; for example, the body atom is \lstinline;p(0); and the fact is \lstinline;p(1);.

When a work item is processed, it computes the direct consequences of the rule it corresponds to, with respect to the \emph{current contents} of relations.
This differs from semi-naive evaluation, which refers to the contents of relations at the beginning of the iteration (i.e., it indexes into relations of the form $p^{[i]}$); eager evaluation does not have explicit iterations and instead works in an unbatched, ``tuple at a time'' fashion.
The fact that rule evaluation uses more up-to-date relations than in semi-naive evaluation does not impact correctness, but it does mean that eager evaluation can materialize some redundant combinations of tuples avoided in semi-naive evaluation.

If a new fact $p(\many{n})$ is derived for some constants $\many{n}$, we add the fact to the $p$ relation and immediately use it to create new work items.
Say there are $k$ occurrences of predicates shaped like $p(\many{t})$ in the bodies of rules that are recursive in the current stratum; the algorithm will try to submit $k$ new work items, one for each occurrence, specializing the body atom $p(\many{t})$ in the relevant rule to the new fact $p(\many{n})$.
Fewer than $k$ work items might be submitted if specialization fails in some case.

\subsubsection{Pseudocode}

\begin{algorithm}[!t]
  \caption{
    Function {\sc workerThread}(i) gives the core logic performed by each worker thread $i$ in eager evaluation.
    For each worker thread $i$, there is a worklist (deque) $D_i$; for each predicate $p$, there is a fact database $DB_p$.
    Worklists and databases must be thread safe.
  }
  \label{algo:eager}
  \begin{algorithmic}[1]
    \Function{workerThread}{$i$}
    \Loop
    \LComment{Dequeue a work item (a rule to evaluate)}
    \State $R \gets D_i.dequeueBack()$
    \Comment{Use LIFO discipline for local worklist}
    \While{$R = \bot \wedge workLeft()$} \Comment{Scan other threads if local worklist is empty}
    \State $j \gets chooseRandomThread()$
    \State $R \gets D_j.dequeueFront()$
    \Comment{Use FIFO discipline if stealing from another thread}
    \EndWhile
    \If{$R = \bot$}
    \Comment{No work left in thread pool}
    \State \textbf{break}
    \EndIf
    \LComment{Evaluate rule $R$, eagerly yielding facts as they are derived (i.e., no batching)}
    \ForAll{$p(\many{n}) \in \Call{evalRule}{R}$}
    \If{$DB_p.addIfAbsent(\many{n})$}
    \Comment{True when tuple $\many{n}$ is novel}
    \ForAll{$R' \in specializeRules(p(\many{n}))$}
    \State $D_i.enqueueBack(R')$
    \Comment{Insert work item in back of current thread's worklist}
    \EndFor
    \EndIf
    \EndFor
    \EndLoop
    \EndFunction
  \end{algorithmic}

  \algrenewcommand\algorithmicreturn{\textbf{yield}}
  \algrenewcommand\algorithmicloop{\dots}
  \begin{algorithmic}[1]
    \setcounter{ALG@line}{14}
    \Statex
    \Function{evalRule}{$p_0(\many{t}_0) \horn p_1(\many{t}_1), p_2(\many{t}_2), \dots, p_n(\many{t}_n).$}
    \State $\sigma_1 = \lambda x. \bot$ \Comment{Start with empty substitution}
    \ForAll{$\many{n}_1 \in DB_{p_1}.query(\sigma_1(\many{t}_1))$}
    \Comment{Iterate through tuples matching key}
    \State $\sigma_2 = \sigma_1 \circ [\many{t}_1 \mapsto \many{n}_1]$
    \Comment{Compose substitutions}
    \ForAll{$\many{n}_2 \in DB_{p_2}.query(\sigma_2(\many{t}_2))$}
    \State $\sigma_3 = \sigma_2 \circ [\many{t}_2 \mapsto \many{n}_2$]
      \Loop
      \ForAll{$\many{n}_n \in DB_{p_n}.query(\sigma_n(\many{t}_n))$}
      \State $\sigma_{n + 1} = \sigma_n \circ [\many{t}_n \mapsto \many{n}_n$]
    \State \Return{$p_0(\sigma_{n + 1}(\many{t}_0))$}
    \Comment{Eagerly yield derived facts (do not batch them)}
    \EndFor
    \EndLoop
    \EndFor
    \EndFor
    \EndFunction
  \end{algorithmic}
\end{algorithm}

Algorithm~\ref{algo:eager} gives pseudocode for the logic implemented by each worker thread $i$ during eager evaluation.
The function {\sc workerThread} gives the core procedure performed by thread $i$, which loops until there is no work left to do (line 8).
In each iteration, thread $i$ first tries to dequeue a work item from the back of its own deque (line 4); if this fails because the deque is empty (the returned work item is $\bot$), it repeatedly tries to steal a work item from the deques of the other threads (lines 5-7), scanning until it successfully steals a work item or the function call $workLeft()$ returns false (indicating that all threads are trying to steal work items, and so all deques are empty).
In practice, the work-stealing thread pool---not code implemented by the Datalog system engineer---distributes work items and decides when to terminate the threads.
Once thread $i$ has a work item in the form of a rule $R$, the thread evaluates rule $R$ using the function {\sc evalRule}; the evaluation code (lines 16-24) is standard for Datalog, except that facts are eagerly yielded as they are derived (line 24).
That is, the function {\sc evalRule} returns a fact {\em generator}, instead of a concrete set of facts.
For each generated fact, thread $i$ checks if that fact is novel (line 12); if so, it adds a work item to the end of its deque for each rule specialized to that fact (lines 13-14).

Because the code in the {\sc evalRule} function is very similar to the code used to evaluate a rule in semi-naive evaluation (cf. Section~\ref{sec:datalog}) and uses the same set of data structure operations (e.g., looping over the index of a relation), we are able to heavily reuse existing Datalog infrastructure when implementing eager evaluation, as we demonstrate in Section~\ref{sec:eagerimpl}.

\subsection{Correctness of Eager Evaluation}\label{sec:eager_correctness}

Eager evaluation is correct; in particular, one can prove the following theorem (see \iffull Appendix~\ref{sec:correctness_proofs} \else \ifanonymous the supplemental material \else the extended version of this paper~\citep{full} \fi \fi for the complete proofs of all theorems):

\begin{theorem}[Correctness]\label{thm:main-correctness}
  Given an arbitrary Datalog program, eager evaluation and semi-naive evaluation derive exactly the same facts.
\end{theorem}

\noindent The proof of Theorem~\ref{thm:main-correctness} rests on two lemmas, both of which make a claim about the behavior of eager evaluation on a given stratum (as opposed to the complete Datalog program).
First, eager evaluation is sound on the stratum: it derives {\em only} the facts derived by semi-naive evaluation (Lemma~\ref{lemma:main-soundness}).
Second, eager evaluation is complete on the stratum: it derives {\em all} the facts derived by semi-naive evaluation (Lemma~\ref{lemma:main-completeness}).

\begin{lemma}[Soundness]\label{lemma:main-soundness}
  Fixing the facts derived in previous strata, if eager evaluation derives a fact $a_0$ in the current stratum, semi-naive evaluation also derives fact $a_0$ in the current stratum.
\end{lemma}

\begin{lemma}[Completeness]\label{lemma:main-completeness}
  Fixing the facts derived in previous strata, if semi-naive evaluation derives a fact $a_0$ in the current stratum, eager evaluation also derives fact $a_0$ in the current stratum.
\end{lemma}

\noindent Completeness is significantly more complex to prove than soundness; in fact, it relies on some (mild) assumptions about the behavior of the concurrent mechanisms used to implement Algorithm~\ref{algo:eager}.

Completeness is challenging to prove because eager evaluation is a concurrent algorithm that uses little explicit synchronization, and so there appears to be the possibility that data races result in the failure to derive a fact.
For example, consider evaluating the rule \lstinline;s(3) :- p(1), q(2);.
Call this rule $R$, and assume that the predicates \lstinline;p; and \lstinline;q; are recursive in the current stratum.
In a multithreaded setting, the derivation of the facts \lstinline;p(1); and \lstinline;q(2); might lead to two work items that are processed concurrently, one specialized to fact \lstinline;p(1); and one specialized to fact \lstinline;q(2);.
Because of data races between the writing and reading of relations, one might worry that \emph{both} work items fail to derive the fact \lstinline;s(3);---i.e., could it be that the derivation of the fact \lstinline;q(2); is not visible in the work item specialized to \lstinline;p(1); and vice versa?
Thankfully, no! If one work item fails because of a data race, the other one necessarily succeeds.
Here, we informally demonstrate why eager evaluation is complete on this example; see \iffull Appendix~\ref{sec:correctness_proofs} \else \ifanonymous the supplemental material \else the extended version of this paper~\citep{full} \fi \fi for a formal proof of the general case.

When facts \lstinline;p(1); and \lstinline;q(2); are derived, these events occur (line numbers refer to Algorithm~\ref{algo:eager}):
\begin{itemize}
  \item Write $w_{1,p}$ (line 12): the fact \lstinline;p(1); is added to data structure $DB_p$.
        Subsequently, Work Item 1 specializing rule $R$ on \lstinline;p(1); is submitted to the thread pool (line 14).
  \item Execution start $x_{1,p}$ (line 4 or 7): Work Item 1 is dequeued to be executed.
  \item Read $r_{1,q}$ (line 19): the data structure $DB_q$ is read when Work Item 1 is processed.
  \item Write $w_{2,q}$ (line 12): the fact \lstinline;q(2); is added to data structure $DB_q$.
        Subsequently, Work Item 2 specializing rule $R$ on \lstinline;q(2); is submitted to the thread pool (line 14).
  \item Execution start $x_{2,q}$ (line 4 or 7): Work Item 2 is dequeued to be executed.
  \item Read $r_{2,p}$ (line 19): the data structure $DB_p$ is read when Work Item 2 is processed.
\end{itemize}
The completeness of eager evaluation rests on a transitive happens-before relationship \happensbeforesym between these events; we require the happens-before relationship to meet the following conditions:
\begin{enumerate}
  \item Reads and writes on the same data structure are ordered: $$(\happensbefore{r_{1,q}}{w_{2,q}} \vee \happensbefore{w_{2,q}}{r_{1,q}}) \wedge (\happensbefore{r_{2,p}}{w_{1,p}} \vee \happensbefore{w_{1,p}}{r_{2,p}}).$$
        This condition would be satisfied by using linearizable concurrent data structures~\cite{Herlihy1990Linearizability}, a common notion of concurrent data structure correctness.
  \item Reads that occur in a work item happen after that work item starts executing: $\happensbefore{x_{1,p}}{r_{1,q}} \wedge \happensbefore{x_{2,q}}{r_{2,p}}$.
        In general, the happens-before relation should respect program order.
  \item Writes happen before the work items they spawn: $\happensbefore{w_{1,p}}{x_{1,p}} \wedge \happensbefore{w_{2,q}}{x_{2,q}}$.
        This condition is a basic visibility requirement.
        In the pseudocode, it would be sufficient if each deque $D_i$ is linearizable; in practice, a thread pool would likely guarantee this condition by performing some form of synchronization when it handles the distribution of work items.
\end{enumerate}
These mild assumptions are reasonable to meet in practice.

To show that fact \lstinline;s(2); is derived, it is enough to show that \happensbefore{w_{1,p}}{r_{2,p}} holds or \happensbefore{w_{2,q}}{r_{1,q}} holds---that is, the write of \lstinline;p(1); is visible in Work Item 2 or the write of \lstinline;q(2); is visible in Work Item 1.
By Assumption (1), we have the disjunction $\happensbefore{r_{1,q}}{w_{2,q}} \vee \happensbefore{w_{2,q}}{r_{1,q}}$.
If the second disjunct is true, Work Item 1 will derive fact \lstinline;s(2);.
Otherwise, we have (read left-to-right, top-to-bottom):
\[
  \begin{array}{lrlrlr}
    \happensbefore{w_{1,p}}{x_{1,p}} & \text{(Assumption (3))} &
    \happensbefore{x_{1,p}}{r_{1,q}} & \text{(Assumption (2))} &
    \happensbefore{r_{1,q}}{w_{2,q}} & \text{(Disjunct case)}    \\
    \happensbefore{w_{2,q}}{x_{2,q}} & \text{(Assumption (3))} &
    \happensbefore{x_{2,q}}{r_{2,p}} & \text{(Assumption (2))}   \\
  \end{array}
\]
Thus, by transitivity, we have that \happensbefore{w_{1,p}}{r_{2,p}} holds, and Work Item 2 will derive fact \lstinline;s(2);.

\section{Eager Evaluation in Practice}\label{sec:eagerimpl}

This section discusses our experience adding eager evaluation to the Formulog interpreter (Section~\ref{sec:eagerevalinterpret}) and \souffle's code generator (Section~\ref{sec:eagerevalsouffle}), and shows that eager evaluation is a practical and effective strategy for evaluating Formulog programs (Section~\ref{sec:eagerevalresults}): the eager evaluation mode of the \emph{interpreter} achieves a mean speedup of 5.2$\times$ over the semi-naive code generated by \souffle on SMT-heavy benchmarks, and the eager evaluation extension to (compiled) \souffle achieves a mean 1.8$\times$ speedup over the Formulog interpreter's eager evaluation mode and a mean 7.6$\times$ speedup over off-the-shelf \souffle.
Thus, Formulog is an example of a Datalog variant that benefits from a non-standard evaluation technique; moreover, this technique can be built on top of existing Datalog infrastructure with relatively low effort.

\subsection{Extending the Formulog Interpreter}\label{sec:eagerevalinterpret}

Our implementation of eager evaluation in the Formulog interpreter consists of two parts: a short method that reorders the bodies of semi-naive rules so that $\delta$-predicates come first,%
\footnote{
  Our implementations {\em lazily} specialize rules to new facts---i.e., during work item evaluation, not during work item creation (as in Algorithm~\ref{algo:eager}).
  Recall that $\delta$-relations are auxiliary relations used to focus traditional semi-naive evaluation on newly derived tuples (Section~\ref{sec:datalog}).
  When a work item is processed, the given rule is run on a singleton, mock $\delta$-relation containing just the tuple to specialize on.
  Since the mock $\delta$-predicate is the first atom, this has the effect of specializing the rule.
}
and a class that evaluates a stratum of the Formulog program using eager evaluation (200 lines of Java).
The classes for eager evaluation and semi-naive evaluation both descend from a shared abstract class (200 lines); both use a work-stealing thread pool (a \lstinline;java.util.concurrent.ForkJoinPool; instance) to handle parallelism; and both store relations in the same indexed data structures built around \lstinline;java.util.concurrent.ConcurrentSkipListSet; instances.
The remaining logic (and code) in each class is also very similar, with the main difference occurring when a new tuple is derived: in eager evaluation, new work items are submitted corresponding to rules specialized on that tuple; in semi-naive evaluation, the new tuple is added to an auxiliary relation.

\subsection{Extending \souffle's Code Generator}\label{sec:eagerevalsouffle}

Our extension to Souffle's code generator consists of three parts: a short class for reordering semi-naive rule bodies so that $\delta$-predicates are first, modifications to the code generator for semi-naive rules, and modifications to the code generator for the data structures representing relations.
The latter two occur during the translation from relational algebra machine (RAM) instructions to \cc.
RAM instructions capture high-level imperative operations, like scanning the index of a data structure, or inserting a tuple into a relation.
To support eager evaluation, we change the translation of some RAM instructions.
For example, we change insert instructions to insert a new tuple into the base relation (instead of an auxiliary relation), and to submit new work items representing rules specialized on that tuple.
We also modify code generation so that the logic for evaluating a rule is wrapped up in its own function; consequently, a work item's single operation is to invoke a ``rule function'' on the new tuple.
Work items are submitted to a \lstinline;task_group; instance from Intel's oneAPI Thread Building Blocks (oneTBB) concurrency library%
\footnote{\url{https://www.intel.com/content/www/us/en/developer/tools/oneapi/onetbb.html}}%
; this uses a work-stealing thread pool on the backend to process tasks.

In addition to modifying code generation for rule evaluation, we modify code generation for the data structures representing relations.
This is required because of a limitation in \souffle's concurrent data structures, which do not support concurrent reads and writes.
Because there are distinct reading and writing phases in semi-naive evaluation, the \souffle implementors have added only the synchronization necessary for multiple reads happening simultaneously or multiple writes happening simultaneously.
To get around this, we use instances of oneTBB's \lstinline;concurrent_set; class for indices where necessary.
These sets are based on concurrent skip lists and are much slower than \souffle's specialized B-trees and tries.
Accordingly, we use oneTBB's concurrent sets for an index only if it is accessed when evaluating a non-``delta'' recursive predicate in a rule body (a necessary condition for an index being read and written concurrently).
Programs with only linearly recursive rules can thus get away with using just \souffle's data structures.

All in all, our modifications required adding just 500 lines of \cc to \souffle (which is 70k lines of \cc, including code not involved in compilation).
The eager evaluation code generator can in principle be used to compile arbitrary Datalog programs, and is agnostic to the fact that we use it as part of a compiler for Formulog programs.
However, our extension does not currently support all of \souffle's RAM instructions (we focus on those used during our case studies), nor all of \souffle's additional machinery (e.g., provenance tracking).
To link with code generated for eager evaluation, we made minor modifications to the Formulog \cc runtime (for example, using oneTBB's abstraction for thread-local variables instead of OpenMP's).

\subsection{Evaluation}\label{sec:eagerevalresults}

\begin{figure}[t]
  \includegraphics{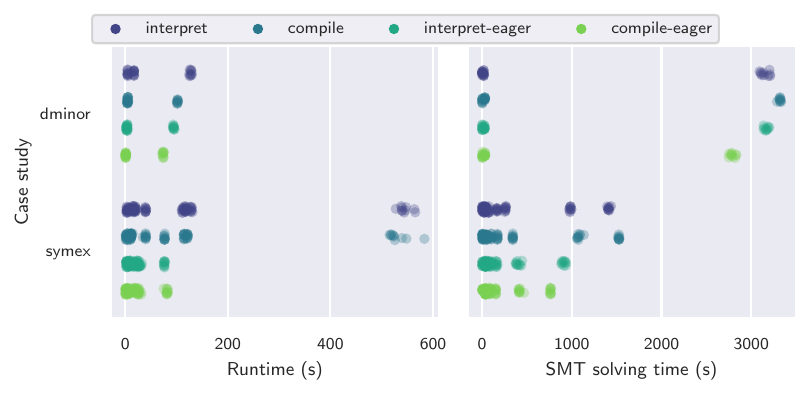}
  \caption{
    Eager evaluation achieves speedups on SMT heavy benchmarks over semi-naive evaluation, primarily by improving SMT solving times (eager evaluation's finer-grained parallelism also helps to a small degree).
    Note that the \emph{CPU time} of SMT solving (as shown in the right plot) does not capture the full extent to which eager evaluation's redistribution of SMT work can lead to lower SMT solving \emph{wall-clock time} (e.g., eager evaluation might induce two SMT queries to be solved in parallel instead of sequentially; the CPU time spent in SMT solving would be the same, but the wall-clock time would be less).
  }
  \label{fig:eager_eval}
\end{figure}

We first summarize the experimental results across the benchmark suite (Section~\ref{sec:eager_results_summary}), and then detail the results for each case study (Sections~\ref{sec:eager_dminor} and~\ref{sec:eager_symex}).
We also evaluate how well eager evaluation scales with the number of threads compared to semi-naive evaluation (Section~\ref{sec:scaling}).

\subsubsection{Summary}\label{sec:eager_results_summary}
Eager evaluation leads to speedups on SMT-heavy benchmarks (Figure~\ref{fig:eager_eval}).
In particular, on the \bm{dminor} refinement type checker and the \bm{symex} symbolic executor, the interpreter's eager evaluation mode (``\bm{interpret-eager}'') achieves a 5.2$\times$
mean speedup (min\sep{}median\sep{}max: \speedups{0.41}{1.2}{31}) over the code generated by our initial compiler (``\bm{compile}''); in turn, using the eager evaluation extension to \souffle (``\bm{compile-eager}'') achieves a 1.8$\times$ mean speedup (min\sep{}median\sep{}max: \speedups{0.93}{1.3}{4.1}) over the \bm{interpret-eager} mode.
This section omits results for the \bm{scuba} points-to analysis case study, where eager evaluation is not remotely competitive with our initial compiler to off-the-shelf \souffle.
The \bm{scuba} case study uses little SMT solving and is most similar to traditional Datalog workloads; additionally, eager evaluation reorders rules so that $\delta$-predicates come first, which happens to be a very inefficient join ordering strategy here.

Eager evaluation uses more memory than semi-naive evaluation:
\bm{interpret-eager} uses mean 1.6$\times$ as much memory as the baseline interpreter (``\bm{interpret}'') (min\sep{}median\sep{}max: \speedups{0.96}{1.1}{5.2});
\bm{compile-eager} uses mean 1.1$\times$ as much memory as \bm{compile} (min\sep{}median\sep{}max: \speedups{0.92}{1.1}{1.5}).

Eager evaluation's approach to parallelism is finer-grained than the parallel \lstinline;for; loops of typical implementations of semi-naive evaluation, and in general the eager evaluation modes demonstrate higher CPU utilization in our experiments (albeit with some variability).
\bm{Interpret-eager} uses mean 3.9$\times$ as many CPUs as \bm{interpret} (min\sep{}median\sep{}max: \speedups{0.93}{1.9}{12}); \bm{compile-eager} uses mean 1.6$\times$ as many CPUs as \bm{compile} (min\sep{}median\sep{}max: \speedups{0.19}{1.1}{7.1}).

We also measured the amount of ``work'' performed by each strategy.
We define this quantity as the number of tuple accesses made during the evaluation of a rule; if evaluation takes the form of a nested \lstinline;for; loop where loops iterate over relation contents, it is the total number of iterations across all loops.
Given a fixed program and input, semi-naive evaluation will always perform the same amount of work (which will, in general, contain some redundant combinations of tuples).
Not so for eager evaluation, where the outcome of data races can lead to a different amount of work being performed (which might happen to be more or less than the work done by semi-naive evaluation).
In our experiments, \bm{interpret} and \bm{interpret-eager} do almost exactly the same amount of work, but \bm{compile} does slightly more work (1.1-1.2$\times$) than \bm{compile-eager}.%
\footnote{For the \bm{dminor} and \bm{symex} case studies, \bm{interpret} and \bm{compile} modes both use a $\delta$-predicate-first heuristic for reordering atoms in rule bodies, as that ordering is more efficient on these benchmarks than the default left-to-right order. Since this is consistent with the atom ordering used by eager evaluation, comparing the amount of work done by semi-naive and eager modes is apples-to-apples (in general, the amount of work done in a rule depends on the order of body atoms).}

\subsubsection{Refinement Type Checking (\bm{dminor})}\label{sec:eager_dminor}

\bm{Interpret-eager} has a mean 1.3$\times$ speedup (min\sep{}median\sep{}max: \speedups{1.1}{1.3}{1.7}) over \bm{compile} on the \bm{dminor} case study, while \bm{compile-eager} has a mean 2.8$\times$ speedup (min\sep{}median\sep{}max: \speedups{1.3}{3.0}{4.1}) over \bm{interpret-eager}.

Eager evaluation is not universally more effective than \bm{compile} at grouping similar SMT calls on the same thread: \bm{compile} has mean 1.2$\times$ (min\sep{}median\sep{}max: \speedups{0.97}{1.2}{1.3}) as many SMT cache misses as \bm{interpret-eager}, but 0.94$\times$ (min\sep{}median\sep{}max: \speedups{0.80}{1.0}{1.0}) as many as \bm{compile-eager}.%
\footnote{
  The Formulog runtime uses encoding tricks to maintain a per-solver cache of conjuncts~\citep{formulog_incremental_smt}; we count an SMT cache miss for each relevant conjunct that is not in the cache when an SMT call is made.
}
This suggests that SMT calls in this case study do not follow the pattern identified in Section~\ref{sec:eagereval}.
Despite this, eager evaluation can still lead to better SMT solving CPU times than \bm{compile}: SMT solving can be up to 1.2$\times$ faster with both \bm{interpret-eager} and \bm{compile-eager}.
Thus, eager evaluation leads to an advantageous distribution of the SMT workload, despite taking less advantage of SMT caching.
This perhaps reflects the difficulty of predicting SMT solver performance, which can vary widely with small perturbations to solver state and query framing (e.g., the choice of SMT variable names can affect how long solving takes and whether the result is \lstinline;sat; or \lstinline;unknown;%
\footnote{See \url{https://github.com/Z3Prover/z3/issues/909/} and \url{https://github.com/Z3Prover/z3/issues/4600}.}%
).

\subsubsection{Symbolic Execution (\bm{symex})}\label{sec:eager_symex}

\bm{Interpret-eager} achieves a mean 6.3$\times$ speedup (min\sep{}median\sep{}max: \speedups{0.4}{1.1}{31}) over \bm{compile} on the \bm{symex} benchmarks, while \bm{compile-eager} manages a 1.5$\times$ speedup (min\sep{}median\sep{}max: \speedups{0.9}{1.2}{3.5}) over \bm{interpret-eager}.

The SMT story is more straightforward for this case study: eager evaluation more effectively groups together similar SMT calls on the same thread, and this leads to faster SMT solving times.
In particular, \bm{compile} has mean 2.5$\times$ (min\sep{}median\sep{}max: \speedups{0.34}{1.4}{5.7}) as many SMT cache misses as \bm{interpret-eager} and mean 5.7$\times$ (min\sep{}median\sep{}max: \speedups{0.3}{3.6}{18}) as many as \bm{compile-eager}.
In terms of CPU time spent SMT solving, this gives a mean speedup of 1.4$\times$ over \bm{compile} for both \bm{interpret-eager} (min\sep{}median\sep{}max: \speedups{0.28}{1.5}{2.7}) and \bm{compile-eager} (min\sep{}median\sep{}max: \speedups{0.24}{1.6}{2.8}).
Thus, as anticipated in Section~\ref{sec:eagereval}, symbolic execution discharges SMT queries in a pattern well matched to eager evaluation.

\begin{figure}[t]
  \includegraphics{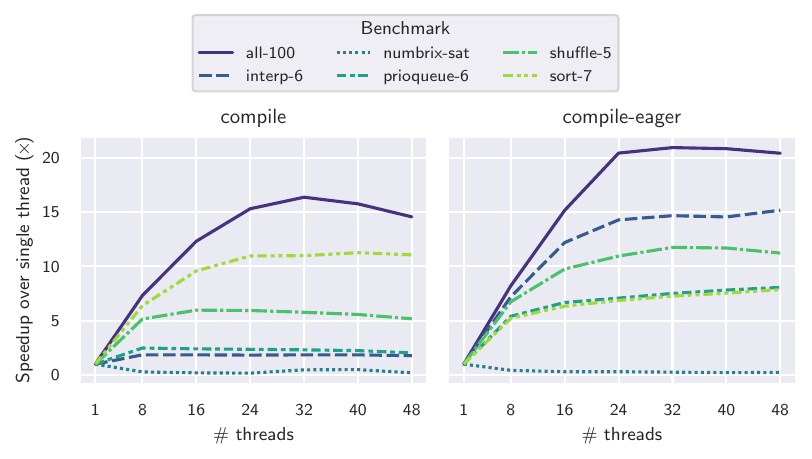}
  \caption{Eager evaluation typically takes better advantage of additional threads than semi-naive evaluation, achieving higher speedups (relative to single-threaded evaluation) for a given thread count, and continuing to improve performance past the thread count where semi-naive evaluation's performance stalls or decreases.}
  \label{fig:scaling}
\end{figure}

\subsubsection{Scaling}\label{sec:scaling}

We also tested how well eager evaluation and semi-naive evaluation scale as the number of threads is varied.
We ran the compile and compile-eager modes on six benchmarks: the longest-running \bm{dminor} benchmark (\bm{all-100}), and the longest-running configuration of each of the five template \bm{symex} benchmarks (\bm{interp-6}, \bm{numbrix-sat}, \bm{prioqueue-6}, \bm{shuffle-5}, and \bm{sort-7}).%
\footnote{We omit the interpreted Formulog modes on the basis that they are less optimized than the compiled modes, and so how well they scale is a less interesting question.}
We varied the number of threads between 1, 8, 16, 24, 32, 40, and 48, with five trials per configuration.
These experiments were run on an Ubuntu 22.04 c5.12xlarge Amazon Web Services EC2 with an Intel Xeon processor clocked at \SI{3.6}{GHz}, 24 physical cores (48 vCPUs), and \SI{96}{GiB} RAM.

At the baseline of using a single thread, compile-eager outperforms compile by a mean speedup of 1.7$\times$ (min\sep{}median\sep{}max: \speedups{0.9}{1.2}{4.3}).
This speedup is slightly less than the mean 2.1$\times$ speedup compile-eager has over compile in terms of SMT solving times (min\sep{}median\sep{}max: \speedups{0.9}{1.5}{5.2}), an indication of the overhead of eager evaluation relative to semi-naive evaluation.

On most benchmarks, compile-eager scales better than semi-naive evaluation, achieving higher speedups over single-threaded evaluation as the number of threads is varied (Figure~\ref{fig:scaling}).
Compile-eager is typically able to more fully take advantage of additional threads, achieving higher speedups than compile for a given thread count, and continuing to get performance improvements past the thread count where compile's performance begins to stall or decrease.
Nonetheless, most of compile-eager's performance is achieved by 32 threads, with diminishing returns beyond that point.

Two benchmarks are outliers to the general trend.
First, single-threaded evaluation is best for both evaluation modes on the symex benchmark \bm{numbrix-sat}, which has a single possible program path: additional threads lead to more overhead, with no performance gain (there are no branches for symbolic execution to explore in parallel).
Second, on the \bm{symex} benchmark \bm{sort-7}, compile achieves higher speedups than compile-eager.
On this benchmark, compile-eager's SMT solving CPU times are about twice as fast as compile's; nonetheless, its overall times are slightly slower.
This benchmark generates many (28 million) derived tuples, and it could be that eager evaluation's approach of generating one work item per tuple introduces too much overhead.
This result suggests the potential of a hybrid evaluation scheme that uses some amount of batching to reduce overhead, while still achieving better SMT solver locality than semi-naive evaluation.

\definecolor{Gray}{gray}{0.9}

\begin{table}
  \caption{
    On 22/23 benchmarks, a compiled Formulog mode beats any Formulog interpreter mode; on 20/23 benchmarks, some Formulog mode beats the non-Datalog reference implementation (see Section~\ref{sec:ref_compare}; for the \bm{symex} case study, we compare against KLEE~\citep{klee}).
    The best Formulog time on a benchmark is \textit{italicized}; the best overall time for a benchmark is in \textbf{bold}. Timeouts were 30 minutes.
    In the case that a Formulog configuration timed out on its first trial, we did not run additional trials.
    SMT-heavy case studies are shaded \colorbox{Gray}{gray}.
  }\label{tab:results}
  \begin{tabular}{ll rr rrr rr}
    \toprule
               &               & \multicolumn{2}{c}{Formulog interp. (s)} &                   & \multicolumn{2}{c}{Formulog compile (s)} &                     & Reference                           \\
    \cline{3-4} \cline{6-7}
    Case study & Benchmark     & semi-naive                               & eager             &                                          & semi-naive          & eager             &  & impl. (s)    \\
    \midrule
    \rowcolor{Gray}
    dminor     & all-1         & 4.76                                     & 2.46              &                                          & 4.10                & {\bf {\em 0.60}}  &  & 1.50         \\
    \rowcolor{Gray}
    dminor     & all-10        & 17.05                                    & 4.19              &                                          & 5.38                & {\bf {\em 1.40}}  &  & 68.00        \\
    \rowcolor{Gray}
    dminor     & all-100       & 128.00                                   & 94.56             &                                          & 102.02              & {\bf {\em 74.05}} &  & TO           \\
    scuba      & antlr         & 310.38                                   & TO                &                                          & {\bf {\em 93.89}}   & 1015.10           &  & 95.11        \\
    scuba      & avrora        & 387.19                                   & TO                &                                          & {\em 137.79}        & TO                &  & {\bf 85.16}  \\
    scuba      & hedc          & 249.55                                   & TO                &                                          & {\bf {\em 65.00}}   & 272.36            &  & 81.01        \\
    scuba      & hsqldb        & 244.84                                   & TO                &                                          & {\bf {\em 64.38}}   & 431.53            &  & 68.37        \\
    scuba      & luindex       & 1338.49                                  & TO                &                                          & {\bf {\em 1045.64}} & TO                &  & error        \\
    scuba      & polyglot      & 243.34                                   & 1246.84           &                                          & {\bf {\em 56.66}}   & 164.78            &  & 75.15        \\
    scuba      & sunflow       & 1041.81                                  & TO                &                                          & {\em 495.98}        & TO                &  & {\bf 279.81} \\
    scuba      & toba-s        & 257.02                                   & TO                &                                          & {\bf {\em 75.21}}   & 312.01            &  & 76.63        \\
    scuba      & weblech       & 286.23                                   & TO                &                                          & {\bf {\em 77.66}}   & 799.52            &  & 97.22        \\
    scuba      & xalan         & 285.88                                   & TO                &                                          & {\em 84.56}         & 1057.24           &  & {\bf 72.64}  \\
    \rowcolor{Gray}
    symex      & interp-5      & 117.64                                   & 6.80              &                                          & 121.24              & {\bf {\em 4.18}}  &  & 46.81        \\
    \rowcolor{Gray}
    symex      & interp-6      & 541.73                                   & 16.93             &                                          & 523.55              & {\bf {\em 14.23}} &  & 176.50       \\
    \rowcolor{Gray}
    symex      & numbrix-sat   & 19.84                                    & 28.72             &                                          & {\bf {\em 11.87}}   & 26.79             &  & 302.25       \\
    \rowcolor{Gray}
    symex      & numbrix-unsat & 16.79                                    & 23.02             &                                          & {\bf {\em 10.13}}   & 22.84             &  & 160.03       \\
    \rowcolor{Gray}
    symex      & prioqueue-5   & 39.72                                    & 7.89              &                                          & 39.81               & {\bf {\em 5.33}}  &  & 41.11        \\
    \rowcolor{Gray}
    symex      & prioqueue-6   & 128.35                                   & 24.48             &                                          & 114.72              & {\bf {\em 19.91}} &  & 199.34       \\
    \rowcolor{Gray}
    symex      & shuffle-4     & 3.06                                     & 2.25              &                                          & 1.62                & {\bf {\em 0.64}}  &  & 49.80        \\
    \rowcolor{Gray}
    symex      & shuffle-5     & 6.30                                     & 3.79              &                                          & 4.46                & {\bf {\em 2.06}}  &  & TO           \\
    \rowcolor{Gray}
    symex      & sort-6        & 15.04                                    & 9.62              &                                          & 10.80               & {\bf {\em 9.11}}  &  & 140.20       \\
    \rowcolor{Gray}
    symex      & sort-7        & 112.56                                   & {\bf {\em 76.25}} &                                          & 76.77               & 82.02             &  & 1607.07      \\
    \bottomrule
  \end{tabular}
\end{table}

\section{Formulog Performance: Present and Future}\label{sec:discussion}

This section first contextualizes the performance of new-and-improved Formulog relative to non-Datalog reference implementations (Section~\ref{sec:ref_compare}), and then proposes future directions (Section~\ref{sec:future}).

\subsection{Putting Formulog's Performance in Context}\label{sec:ref_compare}

In the original Formulog paper,~\citet{formulog} compare the performance of Formulog to reference implementations that perform analyses similar to the case studies, but are written in traditional languages (F$^\sharp$, Java, and \cc).
While these comparisons are not apples-to-apples (implementations might use different heuristics, SMT encodings, etc.), they do provide wider context for Formulog's performance.
How does Formulog stack up now given compilation and eager evaluation?
The numbers are promising: our performance improvements to Formulog substantially close the gap in cases where Formulog previously significantly trailed the reference implementations, and translate to increased speedups over the reference implementations in other cases (Table~\ref{tab:results}).

Except for the Dminor reference implementation, we ran the reference implementations on an m5.12xlarge AWS EC2 instance, using the same version of Z3 as in our Formulog experiments.
The Dminor reference implementation can run only on an old version of the .NET platform, constrained both in terms of available memory and available Z3 versions; instead of running the Dminor reference implementation again, we use the numbers provided by~\citet{formulog}.

\paragraph{Refinement Type Checking (\bm{dminor})}

The \bm{dminor} case study has three benchmarks, consisting, respectively, of one, 10, and 100 copies of the same skeleton Dminor program $D_{all}$ (a composite of all publicly available Dminor programs that use only the core feature set).
The reference type checker of~\citet{dminor} is written in F$^\sharp$ and uses an optimization not implemented in the Formulog version; nonetheless, Formulog \bm{compile-eager} consistently beats it.
The reference implementation takes 1.5 seconds to type check one copy of $D_{all}$; in \bm{interpret} mode, Formulog takes 4.8 seconds; in \bm{compile-eager} mode, it takes 0.61 seconds.
For the 10-copy version, the reference implementation takes 68 seconds, Formulog \bm{interpret} takes 17 seconds, and Formulog \bm{compile-eager} takes 1.4 seconds.
The Formulog versions scale to type checking 100 copies of $D_{all}$ (\bm{interpret}: 128 seconds; \bm{compile-eager}: 74 seconds); the reference implementation does not, timing out after 100 minutes (perhaps hamstrung by the old .NET platform it runs on).

\paragraph{Bottom-Up Points-To Analysis (\bm{scuba})}

The reference implementation of~\citet{bottomupPointsto}, written in Java, achieves a mean $3.5\times$ speedup over Formulog \bm{interpret} (min\sep{}median\sep{}max: \speedups{2.9}{3.4}{4.5}) on nine substantial Java benchmarks.%
\footnote{We omit the benchmark \bm{luindex}, on which the reference implementation of~\citet{bottomupPointsto} fails with an out-of-bounds access.
  \citet{formulog} report that the reference implementation completes on this benchmark in 100 seconds.
  The experimental setups differ in the implementation of the JDK 7 that is analyzed as part of the points-to analyses; we believe this leads to the discrepancy here.
}
By compiling the Formulog version to \souffle, we reduce this speedup to mean $1.1\times$ (min\sep{}median\sep{}max: \speedups{0.75}{0.98}{1.8}); that is, Formulog \bm{compile} performs essentially on par with the reference implementation.
The Formulog version closely follows the mathematical specification of the analysis, while integrating some of the heuristics used in the reference implementation, which is significantly more complex (and 10$\times$ larger in SLOC).

\paragraph{Symbolic Execution (\bm{symex})}

We compare against the symbolic executor KLEE (v3.0, released June 2023)~\citep{klee}, set to analyze LLVM 13.0 bitcode (the Formulog-based bounded symbolic executor analyzes LLVM 7.0 bitcode).
Formulog \bm{interpret} achieves a mean $35\times$ speedup over KLEE (min\sep{}median\sep{}max: \speedups{0.33}{9.4}{290}), an improved performance compared to the numbers reported in the original Formulog paper.
Formulog \bm{compile-eager} boosts this speedup up to mean $100\times$ (min\sep{}median\sep{}max: \speedups{7.0}{12}{870}).
These numbers should be taken with a grain of salt: the Formulog-based symbolic executor only handles a small subset of LLVM, whereas KLEE is an industrial-strength tool.
Although it implements a different analysis---C bounded model checking instead of LLVM symbolic execution---for the sake of completeness we also compare against the tool CBMC (v5.85.0, released June 2023)~\citep{cbmc}: it achieves a mean $200\times$ speedup over Formulog \bm{compile-eager} (min\sep{}median\sep{}max: \speedups{0.26}{9.2}{890}).
We explicitly set the number of times to unwind loops on two benchmarks to force CBMC to complete.

\paragraph{Discussion}
The performance improvements developed in this paper were enabled by Formulog's high-level design, a central desideratum being that additional features---functional programming, SMT solving, and relation reification---should not get in the way of evaluating and optimizing Formulog programs just as if they were Datalog programs.
First, this allows a direct translation from Formulog to a high-performance Datalog platform, \souffle.
In the future, we could translate to different Datalog engines (e.g., Differential Datalog~\citep{ryzhyk2019}), as long as they have an FFI and provide some control over the order in which predicates in a rule are evaluated (well typedness in Formulog depends on predicate order; our implementations of eager evaluation account for this).
Second, it means that we can develop new \emph{Formulog} evaluation algorithms by developing new \emph{Datalog} evaluation strategies; e.g., by developing eager Datalog evaluation, we get eager Formulog evaluation.
Compared to designing Formulog evaluation algorithms directly, this has the advantage that Datalog has fewer moving parts (and is thus easier to reason about) and is extensively studied.
It also enabled us to relatively easily add eager evaluation to \souffle, as we were able to reuse much of the \souffle codebase, and did not have to make \souffle aware of Formulog semantics.

\subsection{Future Directions}\label{sec:future}

\paragraph{Hybrid Evaluation Schemes}
There is a spectrum between semi-naive and eager evaluation strategies.
On the one hand, it is possible to make eager evaluation a little less eager by introducing some degree of batching; for example, the facts derived during the evaluation of a rule could always be grouped together into fixed-size batches.
On the other hand, it is also possible to use eager evaluation and semi-naive evaluation within the same program; for instance, eager evaluation could be used for Formulog strata that contain SMT calls, and semi-naive evaluation could be used elsewhere.
Furthermore, although it would require a bit more coordination, eager and semi-naive evaluation can be combined in the same stratum, so that different approaches are used for different rules.
Hybrid evaluation schemes would be particularly relevant when integrating a component written in Formulog (and compiled to \souffle) into an existing \souffle project: one might want to use eager evaluation just for the Formulog part, and semi-naive evaluation everywhere else.

\paragraph{More Sophisticated Distribution of SMT Queries}

Formulog currently associates a single SMT solver with each Datalog evaluation thread, and discharges all SMT queries arising on that thread to that solver.
This tightly binds the distribution of SMT queries across SMT solvers to the distribution of inference tasks across Datalog evaluation threads.
While this makes it possible to achieve more favorable distributions of the SMT workload just by changing the Datalog evaluation algorithm (viz. eager evaluation), it is also a limitation, as it rules out other possible ways of distributing the SMT workload.
For example, one could imagine having a \emph{shared} pool of SMT solvers: each Datalog worker thread submits queries to the pool, and the pool decides the ``best'' solver to run it on (perhaps based on the state of the solver, or some characteristic of the query).
The pool might choose to allocate multiple cores/solvers to particularly ``hard'' queries using a portfolio approach~\citep{portfolio_smt}.
However, even with more sophisticated schemes for distributing SMT queries, eager evaluation would still be useful, since it is ultimately the Datalog evaluation strategy that determines the order in which SMT queries are produced.

\paragraph{Applying Eager Evaluation Beyond Formulog.}

Eager evaluation's DFS of the logical inference space could be beneficial to Datalog systems (beyond Formulog) that interact with stateful external systems whose performance depends on the distribution and order of queries.
That interaction might be through a general foreign function interface or via specialized, built-in support.
For example, \souffle supports relations stored in external SQLite databases.
The current implementation of \souffle eagerly reads the entire relation into memory during initialization, but an alternate implementation might read only a slice of the database relation into memory at any one time (e.g., if the relation is too large to fit into memory); in this case, performance depends on the temporal and spatial locality of database queries.
Another example is Vieira~\citep{Li2024Relational}, a probabilistic variant of Datalog with built-in predicates that discharge queries to external foundation models.
Vieira currently interacts with foundation models using only one-off (context-free) calls; in an alternate model of neurosymbolic programming, Vieira could use stateful APIs for foundation models (such as OpenAI's Assistants API%
\footnote{\url{https://platform.openai.com/docs/assistants/overview}}%
) so that each foundation model query is part of the same conversation.
In this case, the way a foundation model answers a question would depend on the flow of the conversation up to that point---a function of how the Datalog inference space has been explored.

\section{Related Work}

\paragraph{Pipelined Semi-Naive Evaluation}
Eager evaluation is similar in spirit to pipelined semi-naive evaluation (PSN), introduced by \citet{declarativenetworking} for Datalog-based declarative networking.
Eager evaluation and PSN are both worklist-based algorithms that evaluate rules using one ``delta'' tuple at a time.
However, PSN is designed for a distributed, multi-node setting, whereas eager evaluation is designed for a single-node, multi-threaded setting.
Despite its distributed setting, PSN is crafted so that rule evaluation is localized, making it easy to prove the algorithm's completeness; in contrast, eager evaluation does not restrict which thread a computation is performed on, leading to data races and a subtle completeness argument.
PSN timestamps facts to avoid doing redundant work compared to traditional semi-naive evaluation; eager evaluation could do something similar, although we suspect this would lead to unacceptable overhead on the workloads Formulog targets.

\paragraph{Executing Datalog Variants}
Non-distributed Datalog variants typically strive to work with semi-naive evaluation~\citep{flix,datafunseminaive,formulog,szabo2021incremental,ascent,Sahebolamri2023Bring}.
Because semi-naive evaluation is the most widely used evaluation strategy for Datalog, being able to support semi-naive evaluation is seen as validating the practicality of the variant (e.g., there is an entire POPL paper on semi-naive evaluation for Datafun~\citep{datafunseminaive}).
Our work demonstrates that, even in a single-node setting, semi-naive evaluation is not the optimal evaluation strategy for every Datalog variant.

\paragraph{Traversal Orders in Logic Programming}
We are not aware of prior work that performs a DFS of the logical inference space during bottom-up evaluation (except for PSN~\citep{declarativenetworking}, which has this capability).
However, top-down logic programming systems like Prolog typically explore SLD resolution trees using DFS~\citep{gallier1986logic}.
Top-down evaluation proceeds from the head of a rule to the body---i.e., starts with a conclusion, and then searches for a justification.
Bottom-up evaluation works the opposite way by combining known facts to make new inferences.
Given a query to prove, the magic set transformation~\citep{magicsets1,magicsets2} rewrites a logic program so that bottom-up evaluation derives only facts relevant to proving that query (simulating top-down evaluation); while this limits which logical inferences are made, those inferences are still made in a breadth-first order (assuming traditional semi-naive evaluation).
In principle, bottom-up evaluation algorithms could pursue traversals through the logical inference space different from breadth-first and depth-first search.
For example, an evaluation strategy that assigns a priority to each derived fact, and processes facts in this order, could be used to guide potentially non-terminating computations---like unbounded symbolic execution~\citep{klee} or enumerative program synthesis~\citep{alur2013syntax}---towards fruitful directions.

\paragraph{Datalog Compilation}

Compilation is a popular strategy for speeding up Datalog and its variants.
\souffle~\citep{souffle1,souffle2} spearheaded this approach by compiling Datalog to an abstract instruction set that is then compiled to \cc.
The Flix implementation~\citep{flix2022} compiles Horn clauses to an abstract machine; Differential Datalog~\citep{ryzhyk2019}, Crepe~\citep{Zhang2020Crepe}, and Ascent~\citep{ascent} compile them to Rust (the latter two using macro-based approaches).
\citet{pacak2022functional} compile Functional IncA---a functional language interpreted under a non-standard, fixpoint semantics---to Datalog.
Our approach (Section~\ref{sec:compiler}) is hybrid, as we compile Formulog Horn clauses to traditional Datalog, but Formulog functional code directly to \cc.
Flan~\citep{Abeysinghe2024Flan} is a new compilation framework for Datalog that leverages Lightweight Modular Staging~\cite{Rompf2010Lightweight} to flexibly produce high performing code.
As future work, the Flan authors suggest using Flan to construct something akin to a ``compiled Formulog.''
We are excited about this direction, as Flan's flexibility would make it easier to combine Formulog with additional language features (like lattices).
Our work suggests that Flan would need to support eager evaluation to achieve top performance on SMT-heavy workloads.

\section{Summary}

This paper has explored both off-the-shelf technologies and custom techniques to speed up Formulog, a domain-specific language that combines Datalog, SMT solving, and functional programming.
Compiling to off-the-shelf \souffle provides solid (mean 2.2$\times$) speedups; an additional advantage comes via \emph{eager evaluation}, a novel strategy that structures Datalog evaluation as a worklist algorithm and achieves a quasi-depth-first search through the logical inference space by submitting inference tasks to a work-stealing thread pool.
The key insight is that, for Formulog, the order in which you explore logical inferences matters, because different orders lead to different distributions of the SMT workload across threads.
On SMT-heavy workloads, eager evaluation extensions to the Formulog interpreter and \souffle's code generator achieve mean 5.2$\times$ and 7.6$\times$ speedups, respectively, over off-the-shelf \souffle.
Our performance improvements to Formulog help make Formulog-based analyses competitive with previously published, non-Datalog analysis implementations, achieving faster times on 20 out of 23 benchmarks.
Our results provide strong evidence that Formulog can be a realistic platform for SMT-based static analysis, and also add some nuance to the conventional wisdom that fast Datalog evaluation depends on semi-naive evaluation.

\section*{Data-Availability Statement}

\ifanonymous
  See the supplementary material for raw results and data analysis scripts.
  All experimental material (including our extensions to Formulog and \souffle) will be made public upon paper publication.
\else
  An artifact supporting the results of this paper is available on Zenodo~\cite{artifact}.
  It includes our extensions to Formulog and \souffle, the benchmarks we use, the data from our experiments, benchmarking and data-processing scripts, and documentation on how to build on top of our software.
  Additionally, Formulog is available at \url{https://github.com/HarvardPL/formulog}.
\fi

\begin{acks}
  We thank Eric Zhang for his contributions to a preliminary implementation of the Formulog-to-\souffle compiler; we still use some of the code he wrote.
  We thank Eddie Kohler and the anonymous reviewers for helpful feedback on earlier versions of this paper.
  We also belatedly thank Yu Feng and Xinyu Wang for answering our questions about Scuba and giving us access to the reference Scuba implementation when we were writing the original Formulog paper (their assistance was mistakenly omitted from that paper's acknowledgments).
\end{acks}

\bibliographystyle{ACM-Reference-Format}
\bibliography{main}

\iffull
  \clearpage
  \appendix

\section{Translation Details}

Here we fill in some missing details in our compiler from Formulog to \souffle and \cc (Section~\ref{sec:compiler}).

\subsection{Helper Functions}\label{sec:helper_funcs}

We describe the helper functions used in the formal definition of our translation (Figure~\ref{fig:translation}):

\begin{itemize}
  \item $makeFlgRuntime(\many{T})$: instantiate the skeleton Formulog \cc runtime with a set of datatype definitions \many{T}.
        The skeleton runtime contains definitions for built-in Formulog functions (e.g., \lstinline;string_concat;), as well as the SMT-LIB interface.
        However, it needs to be specialized with additional data type information specific to the program being translated.
        For example, this function fills in a \cc \lstinline;switch; statement specifying the arity of each constructor symbol and generates \cc functions for destructing each constructor (Section~\ref{sec:compiling_dtors}); it also provides the serialized SMT-LIB definition of each datatype.
  \item $makeAdminRules(\many{F}~\many{\flg{H}})$: generate auxiliary \souffle relations and Horn clauses that induce the \souffle code generator to act in a way consistent with the semantics of the source Formulog program (given as function definitions \many{F} and rules \flg{\many{H}}).
        For example, because of the reification of relations (Section~\ref{sec:reify_relations}), the \souffle program needs to be stratified in a way consistent with the Formulog program; since \souffle does not see the use of predicates $p(\many{e})$ in functional code, we need to generate auxiliary rules that ensure a correct stratification.
        Our current translation of destructor atoms also requires an auxiliary relation (Section~\ref{sec:compiling_dtors}).
  \item $extractCode(\Gamma)$: Extract and compose together the \cc code in the given context $\Gamma$.
        \[
          \begin{array}{lll}
            extractCode(\cdot)                           & = & \bot                         \\
            extractCode(\flg{f} \mapsto \dl{f}, \Gamma') & = & extractCode(\Gamma')         \\
            extractCode(C, \Gamma')                      & = & C \circ extractCode(\Gamma')
          \end{array}
        \]
        The infix notation $\circ$ denotes code composition and $\bot$ denotes the empty \cc program.
  \item $vars(\flg{H})$: return the set of all variables that syntactically appear in Formulog rule $\flg{H}$.
  \item $freeVars(e)$: return the set of free variables in expression $e$.
\end{itemize}

\subsection{Compiling Destructors}\label{sec:compiling_dtors}

As mentioned in Section~\ref{sec:compiling_dtors_main}, a current limitation in \souffle---the inability to extract the arguments of complex terms returned by external functors---makes it artificially complicated to compile destructor atoms $e \rightarrow c(\many{X})$.
A straightforward translation would take advantage of \souffle's built-in support for complex terms.
In particular, in addition to the terms listed in Figure~\ref{fig:datalog}, \souffle has untagged record terms of the forms $\mathsf{nil}$ (the empty record) and $[t, \dots, t]$; that is, a more complete (but still partial) grammar of \souffle's terms is given by:
\[\begin{array}{@{}l@{}rcl@{}}
    \text{Terms} & ~t & ::= & X \BNFALT n \BNFALT @\cpp{f}(\many{t}) \BNFALT \mathsf{nil} \BNFALT [t, \dots, t]
  \end{array}\]
One should be able to translate a Formulog destructor $e \rightarrow c(X_0, \dots, X_{n-1})$ to a \souffle atom of the form $@dtor_c(t) = [X_0, \dots, X_{n-1}]$, where the term $t$ is the translation of the expression $e$ and the functor call $@dtor_c(t)$ returns the arguments of term $t$ if it has the outermost constructor $c$, and the empty record $\mathsf{nil}$ otherwise.
However, this code leads to an assertion failure within the \souffle code generator: evaluating the unification atom produced by the translation would require extracting the arguments of the complex term returned by the external functor call $@dtor_c(t)$.

To circumvent this limitation, our current approach generates $3 + |\many{X}|$ \souffle atoms for each Formulog destructor $e \rightarrow c(\many{X})$, as described in the \rn{Destruct} rule (reproduced from Figure~\ref{fig:translation}):

\infrule[Destruct]{
\Gamma \vdash \transl{e} \Rightarrow t, \Gamma'
\andalso
\{X,Y,Z\} \cap \Delta = \emptyset
\\
\dl{\many{A}}_{dtor} = [
X = t,
@\mathsf{is\_ctor}_c(X) = Y,
\mathsf{move\_barrier}(Y, Z)
]
\\
\dl{\many{A}}_{assign} = [
X_0 = @\mathsf{nth}(0, X, Z),
\dots,
X_{n-1} = @\mathsf{nth}(n-1, X, Z)
]
}{
\Gamma, \Delta \vdash \transl{e \rightarrow c(X_0, \dots, X_{n-1})} \Rightarrow \dl{\many{A}}_{dtor} \concat \dl{\many{A}}_{assign}, \Gamma', \Delta \cup \{X, Y, Z\}
}

\noindent
The atoms output by the translation perform the following four-step procedure:
\begin{enumerate}
  \item Bind a fresh variable $X$ to the term that the expression $e$ compiles to.
  \item Assign the return value of the external functor call $@\mathsf{is\_ctor}_c(X)$ to a fresh variable $Y$.
        This functor returns \lstinline;1; if the outermost constructor of its argument is $c$ and \lstinline;0; otherwise (the runtime has such a \cc function for each constructor $c$).
  \item Query into an administrative relation using the predicate $\mathsf{move\_barrier}(Y, Z)$ where $Z$ is fresh; this relation contains the single tuple \lstinline;(1, 1);, and so \souffle evaluation proceeds past this point only if $Y$ is \lstinline;1;---i.e., the outermost constructor is indeed $c$.
  \item Create $n$ atoms of the form $X_i = @\mathsf{nth}(i, X, Z)$.
        The $\mathsf{nth}$ functor returns the $i$th argument of $X$, which it treats as a constructed value of arity at least $i + 1$.
        While the $\mathsf{nth}$ functor ignores its third parameter, the argument $Z$ helps induce the right behavior from \souffle's code generator, which otherwise sometimes places the $\mathsf{nth}$ functor calls before evaluation has established that the constructor matches (leading to \cc memory access errors): now, \souffle cannot try to make the $\mathsf{nth}$ functor calls before $Z$ is bound to a value, which happens only by evaluating the predicate $\mathsf{move\_barrier}(Y, Z)$.
\end{enumerate}

\noindent
This slightly clunky translation of destructor atoms will be unnecessary once \souffle improves its support for complex terms; we do not believe it leads to a substantial runtime cost in the meantime.

\section{Proofs for the Correctness of Eager Evaluation}\label{sec:correctness_proofs}

We prove that eager evaluation (Algorithm~\ref{algo:eager}) is correct: it derives exactly the facts that semi-naive evaluation derives.
To do so, we first prove the completeness (Section~\ref{sec:completeness}) and soundness of (Section~\ref{sec:soundness}) of eager evaluation with respect to semi-naive evaluation on a {\em single} stratum.
We then use these lemmas to prove the correctness of eager evaluation on a sequence of strata, i.e., a complete Datalog program (Section~\ref{sec:full_correctness}).

\subsection{Completeness}\label{sec:completeness}

Eager evaluation is a concurrent algorithm, and its completeness (with respect to semi-naive evaluation) relies on the visibility guarantees made by the concurrent data structures and worklists used in it.
We axiomatize these guarantees (Section~\ref{sec:concurrency}) and then prove the stratum-level completeness of eager evaluation (Section~\ref{sec:completeness_proof}).

\subsubsection{Preliminaries}\label{sec:concurrency}

Here, we lay out assumptions, notation, and a few lemmas.
All line numbers refer to the pseudocode for Algorithm~\ref{algo:eager}.

Consider some rule $R$ that has the form $p_0(\many{t}_0) \horn p_1(\many{t}_1), \dots, p_n(\many{t}_n)$ for $n > 0$.
Say that facts $a_i$ for $i \in [n]$ have been derived (we use $[n]$ to denote the set $\{1, \dots, n\}$) and can be used to ground the rule $R$: that is, there exists a substitution $\sigma$ such that $p_0(\sigma(\many{t}_0)) \horn a_1, \dots, a_n.$ is a ground version of the rule $R$.
It follows that each fact $a_i$ has the form $p_i(\many{n}_i)$.
Let $a_0$ denote the fact $p_0(\sigma(\many{t}_0))$.

When each fact $a_i$ is derived, it is written to data structure $DB_{p_i}$ (line 12), an event we denote $w_i$.
This write precedes the submission of a work item representing the specialization of rule $R$ to fact $a_i$; we refer to this work item as work item $i$.
Let $x_i$ be the event when work item $i$ is dequeued to be e\underbar{x}ecuted (line 4 or 7).
During the execution of work item $i$, reads occur of the databases storing the current contents of relations.
Let the event $r$ be one of these reads; we say $\mathsf{item}(r) = i$ to indicate that read $r$ happens during the execution of work item $i$.
If event $r$ is a read of database $DB_{p_j}$, we say that $\mathsf{tgt}(r) = DB_{p_j}$ (e.g., if $r$ is a read that occurs on line 17, $\mathsf{tgt}(r) = DB_{p_1}$).
We use the notation $\mathsf{tuples}(r)$ to denote the set of tuples returned by the read $r$.

Our argument for the completeness of semi-naive evaluation relies on the construction of a happens-before relation \happensbeforesym between events.
In addition to being a partial order, we require four axioms about the behavior of this relation.

\begin{axiom}[Visible-Writes]\label{ax:visibility}
  $\forall i \in [n], r.~\happensbefore{w_i}{r} \wedge \mathsf{tgt}(r) = DB_{p_i} \implies \mathbf{n}_i \in \mathsf{tuples}(r).$
\end{axiom}

\begin{axiom}[Writes-Precede-Dequeues]\label{ax:item_writes}
  $\forall i \in [n].~\happensbefore{w_i}{x_i}.$
\end{axiom}

\begin{axiom}[Dequeues-Precede-Reads]\label{ax:item_reads}
  $\forall i \in [n], r.~\mathsf{item}(r) = i \implies \happensbefore{x_i}{r}.$
\end{axiom}

\begin{axiom}[Linearizable-Relations]\label{ax:rw_ordered}
  $\forall i \in [n], r.~\mathsf{tgt}(r) = DB_{p_i} \implies \happensbefore{r}{w_i} \vee \happensbefore{w_i}{r}.$
\end{axiom}

\noindent
Axiom~\ref{ax:visibility} states that if a read and a write access the same data structure and the write happens-before the read, then the written tuple is in the set of tuples returned by the read; this axiom defines what it means for a write to happen-before a read.
Axiom~\ref{ax:item_writes} states that write $w_i$ that initiates the submission of work item $i$ happens-before work item $i$ is dequeued to be executed, a visibility guarantee we expect the work-stealing thread pool to make (given that some form of synchronization is necessary when the thread pool enqueues and dequeues work items).
Axiom~\ref{ax:item_reads} states that any reads that happen during the execution of a work item happen-after the work item is dequeued; this is the normal requirement that a happens-before relation respects program order.
Axiom~\ref{ax:rw_ordered} states that if a read and write both access the same data structure, then one happens-before the other; this is behavior that we would expect from a linearizable concurrent data structure~\citep{Herlihy1990Linearizability}, a standard notion of correctness for concurrent data structures.

In addition to these axioms, we define the proposition $\mathsf{succeeds}(i)$:
$$\mathsf{succeeds}(i) \triangleq \forall j \in [n], r.~\mathsf{item}(r) = i \wedge \mathsf{tgt}(r) = DB_{p_j} \implies \mathbf{n}_j \in \mathsf{tuples}(r).$$
If the proposition $\mathsf{succeeds}(i)$ holds, then work item $i$ will derive the fact $a_0$---i.e., the ground head $p_0(\sigma(\many{t}_0))$---since all the tuples necessary to derive the fact $a_0$ are visible at the relevant reads.
We use the notation $W_i$ to denote the set of writes that happen-before work item $i$ is dequeued to be executed: $W_i \triangleq \{ w_j \mid \happensbefore{w_j}{x_i} \}$.
We denote the set of all writes $\{w_i \mid i \in [n]\}$ with $\mathcal{W}$.

We first prove that all writes that happen-before a work item is executed are visible to the corresponding reads that happen in that work item.

\begin{lemma}\label{lemma:visibility}
  $\forall i, j \in [n], r.~w_j \in W_i \wedge \mathsf{item}(r) = i \wedge \mathsf{tgt}(r) = DB_{p_j} \implies \mathbf{n}_j \in \mathsf{tuples}(r).$

  \begin{proof}
    We have $\happensbefore{w_j}{x_i}$ and $\happensbefore{x_i}{r}$ (Axiom~\ref{ax:item_reads}); by transitivity of $\happensbeforesym$, we have $\happensbefore{w_j}{r}$.
    By Axiom~\ref{ax:visibility}, we conclude $\mathbf{n}_j \in \mathsf{tuples}(r)$.
  \end{proof}
\end{lemma}

We then prove the key lemma in the argument for the completeness of eager evaluation.
Intuitively, the proof shows that if work item $i$ fails because some write $w_j$ is not visible during the execution of that work item, then it must be that when work item $j$ runs, all the writes visible to work item $i$ are visible, plus the write $w_j$.
If work item $j$ then fails because some write $w_k$ is not visible, then it must be that when work item $k$ runs, all the writes visible to work item $j$ (and thus work item $i$) are visible, plus the write $w_k$.
Following this line of argument, there must be some work item where all writes $w_1, \dots, w_n$ are visible, and this work item will derive fact $a_0$.

\begin{lemma}\label{lemma:key_completeness}
  $\forall m.~\big((\exists i \in [n].~|\mathcal{W} - W_i| \le m) \implies \exists j \in [n].~\mathsf{succeeds}(j)\big).$

  \begin{proof}
    We proceed by induction on $m$ (morally, this is strong induction on the number of writes that are not visible to a given work item).

    \begin{proofcases}
      \item[$(m = 0)$]
      We have that there exists an $i$ such that $|\mathcal{W} - W_i| \le 0$---i.e, $W_i = \mathcal{W}$.
      We then have $\mathsf{succeeds}(i)$ by Lemma~\ref{lemma:visibility}.

      \item[$(m = m' + 1)$]
      Our inductive hypothesis is $$(\exists l \in [n].~|\mathcal{W} - W_l| \le m') \implies \exists j \in [n].~\mathsf{succeeds}(j).$$
      We also have the premise that there exists an $i$ such that $|\mathcal{W} - W_i| \le m' + 1$.
      We need to show that there exists a $j$ such that $\mathsf{succeeds}(j)$.

      Consider the proposition $\mathsf{succeeds}(i)$.
      If it is true, then choose $j = i$.
      If it is not true, then there must exist some $k \in [n], r$ such that: $$\mathsf{item}(r) = i \wedge \mathsf{tgt}(r) = DB_{p_k} \wedge \mathbf{n}_k \not \in \mathsf{tuples}(r).$$
      It must be that $\neg (\happensbefore{w_k}{r})$ (Axiom~\ref{ax:visibility}).
      Consequently, we have:
      \[
        \begin{array}{lr}
          \happensbefore{x_i}{r}   & \text{(Axiom~\ref{ax:item_reads})}  \\
          \happensbefore{r}{w_k}   & \text{(Axiom~\ref{ax:rw_ordered})}  \\
          \happensbefore{w_k}{x_k} & \text{(Axiom~\ref{ax:item_writes})}
        \end{array}
      \]
      By transitivity, $\happensbefore{x_i}{x_k}$ and $W_i \subseteq W_k$.
      Additionally, it must be that $w_k \not \in W_i$ (Lemma~\ref{lemma:visibility}); since $\happensbefore{w_k}{x_k}$ (Axiom~\ref{ax:item_writes}), we have $w_k \in W_k$ and thus $W_i \subset W_k$.

      By assumption, we have $|\mathcal{W} - W_i| \le m' + 1$; it must be that $|\mathcal{W} - W_k| \le m'$, in which case the inductive hypothesis completes the proof. \qedhere
    \end{proofcases}
  \end{proof}
\end{lemma}

Since the premise in the implication in Lemma~\ref{lemma:key_completeness} can be trivially satisfied, it follows that if atoms $a_1, \dots, a_n$ have been derived, fact $a_0$ will also be derived via a specialization of rule $R$.

\begin{lemma}\label{lemma:completeness_simple}
  $\exists j \in [n].~\mathsf{succeeds}(j).$

  \begin{proof}
    Apply Lemma~\ref{lemma:key_completeness}, choosing $m = |\mathcal{W}|$ and $i = 1$.
  \end{proof}
\end{lemma}

\subsubsection{Proof of Completeness}\label{sec:completeness_proof}

\begin{lemma}[Completeness]\label{thm:semi_to_eager}
  Fixing the facts derived in previous strata, if semi-naive evaluation derives a fact $a_0$ in the current stratum, eager evaluation also derives fact $a_0$ in the current stratum.

  \begin{proof}
    Say that semi-naive evaluation uses rule $R$ to derive some fact $a_0$ with derivation tree $T$.
    We proceed by strong induction on the height of derivation trees, using the convention that the height of the derivation tree of a fact produced in a previous stratum is 0.

    In the base case, $height(T) = 1$.
    This means that the body of rule $R$ contains no predicates recursive in the current stratum---there can be no data race.
    Since relations from previous strata are the same for both eager and semi-naive evaluation (by assumption), eager evaluation will derive fact $a_0$ using rule $R$ (during one of the initial work items).

    In the inductive case, $height(T) = k + 1$ for some $k$.
    Let facts $a_1, \dots, a_n$ be the direct children of $a_0$ in the tree $T$---i.e., they are the facts used by rule $R$ to derive $a_0$; necessarily, $height(tree(a_i)) \le k$ for $i \in [1, n]$.
    By the inductive hypothesis, eager evaluation derives each of the facts $a_1, \dots, a_n$; by Lemma~\ref{lemma:completeness_simple}, a work item specializing rule $R$ to some fact $a_i$ will derive fact $a_0$.
  \end{proof}
\end{lemma}

\subsection{Soundness}\label{sec:soundness}

\begin{lemma}[Soundness]\label{thm:eager_to_semi}
  Fixing the facts derived in previous strata, if eager evaluation derives a fact $a_0$ in the current stratum, semi-naive evaluation also derives fact $a_0$ in the current stratum.
\end{lemma}

\begin{proof}
  Say that eager evaluation uses rule $R$ to derive some fact $a_0$ with derivation tree $T$.
  We proceed by strong induction on the height of derivation trees, using the convention that the height of the derivation tree of a fact produced in a previous stratum is 0.

  In the base case, $height(T) = 1$; the body of rule $R$ contains no predicates recursive in the current stratum.
  Since relations from previous strata are the same for both eager and semi-naive evaluation (by assumption), semi-naive evaluation will derive fact $a_0$ when it runs rule $R$.

  In the inductive case, $height(T) = k + 1$ for some $k$.
  Let facts $a_1, \dots, a_n$ be the direct children of $a_0$ in the tree $T$---i.e., they are the facts used by rule $R$ to derive $a_0$; necessarily, $height(tree(a_i)) \le k$ for $i \in [1, n]$.
  By the inductive hypothesis, semi-naive evaluation derives each of the facts $a_1, \dots, a_n$.
  Let $h = max\{ height(T_i) : T_i\text{ is the shortest derivation tree of }a_i\}$; it must be $h \le k$.
  One of the facts $a_1, \dots, a_n$ was derived for the first time in semi-naive iteration $h$, which means that there must be an iteration $h + 1$ because iteration stops only when no new facts are derived.
  Furthermore, all the other facts $a_1, \dots, a_n$ must have been derived going into iteration $h + 1$ (because $h$ is the maximum).
  Therefore, semi-naive evaluation will derive fact $a_0$ in iteration $h + 1$ (using a version of rule $R$).
\end{proof}

\subsection{Correctness}\label{sec:full_correctness}

\begin{corollary}\label{thm:stratum_eager_iff_semi}
  Fixing the facts derived in previous strata, eager evaluation derives exactly the facts semi-naive evaluation derives for the current stratum.
\end{corollary}

\begin{proof}
  This follows by Lemmas~\ref{thm:semi_to_eager} and~\ref{thm:eager_to_semi}.
\end{proof}

\begin{theorem}[Correctness]\label{thm:correctness}
  Given an arbitrary Datalog program stratified as $(P_1, \dots, P_n)$, eager evaluation and semi-naive evaluation derive exactly the same facts.
\end{theorem}

\begin{proof}
  It is sufficient to show that eager evaluation and semi-naive evaluation derive the same set of facts for every stratum $P_i$.
  We proceed by induction on $n$.
  In the base case, $n = 0$ and the program is empty; no facts are derived by either eager evaluation or semi-naive evaluation.
  In the inductive case, $n = k + 1$, and we have that eager evaluation has derived exactly the facts semi-naive evaluation has derived for strata $P_1, \dots, P_k$.
  By Corollary~\ref{thm:stratum_eager_iff_semi}, eager and semi-naive evaluation derive the same facts for stratum $P_{k+1}$.
\end{proof}

\fi

\end{document}